\documentclass[prx,aps,superscriptaddress,footinbib,twocolumn]{revtex4-2}
\renewcommand{\Re}{\operatorname{Re}}

\newcommand{\abs}[1]{\vert #1 \vert}

\usepackage{amsthm,amsmath,amssymb}
\usepackage{mathrsfs}
\newcommand{\T}{\mathrm{T}}
\newcommand{\Tr}{\operatorname{Tr}}
\newcommand{\norm}[1]{\Vert #1 \Vert}

\newcommand{\st}{\,\vert\,}

\newcommand{\ket}[1]{\vert{ #1 }\rangle}
\newcommand{\bra}[1]{\langle{ #1 }\vert}
\newcommand{\ketbra}[2]{\vert #1 \rangle \langle #2 \vert}
\newcommand{\braket}[2]{\langle #1 \vert #2 \rangle}
\newcommand{\mean}[1]{\langle #1 \rangle}

\newcommand{\E}{\operatorname{E}}
\newcommand{\Var}{\operatorname{Var}}

\newcommand{\Relu}{\operatorname{ReLU}}

\newcommand{\calV}{\mathcal{V}}

\usepackage{amsthm}
\newtheorem{theorem}{Theorem}
\newtheorem{proposition}{Proposition}
\newtheorem{lemma}{Lemma}

\theoremstyle{definition}

\theoremstyle{remark}

\usepackage{amssymb}

\usepackage{subfigure}
\usepackage{booktabs}
\usepackage{multirow}
\usepackage{graphicx}
\usepackage{epstopdf}

\usepackage{algorithm}
\usepackage[noend]{algpseudocode}

\usepackage{natbib}

\usepackage[colorlinks=true,citecolor=blue,linkcolor=magenta]{hyperref}

\usepackage{xcolor}
\newcommand{\BLUE}[1]{{\color{blue} #1}}

\usepackage{comment}

\begin{document}




\title{Maximising Quantum-Computing Expressive Power through Randomised Circuits}


\author{Yingli Yang}
\thanks{These three authors contributed equally.}
\affiliation{Institute of Fundamental and Frontier Sciences, University of Electronic Science and Technology of China, Chengdu, 610051, China}

\author{Zongkang Zhang}
\thanks{These three authors contributed equally.}
\affiliation{Graduate School of China Academy of Engineering Physics, Beijing 100193, China}

\author{Anbang Wang}
\thanks{These three authors contributed equally.}
\affiliation{Graduate School of China Academy of Engineering Physics, Beijing 100193, China}

\author{Xiaosi Xu}
\affiliation{Graduate School of China Academy of Engineering Physics, Beijing 100193, China}

\author{Xiaoting Wang}
\email{xiaoting@uestc.edu.cn}
\affiliation{Institute of Fundamental and Frontier Sciences, University of Electronic Science and Technology of China, Chengdu, 610051, China}

\author{Ying Li}
\email{yli@gscaep.ac.cn}
\affiliation{Graduate School of China Academy of Engineering Physics, Beijing 100193, China}

\begin{abstract}
In the noisy intermediate-scale quantum era, variational quantum algorithms (VQAs) have emerged as a promising avenue to obtain quantum advantage. However, the success of VQAs depends on the expressive power of parameterised quantum circuits, which is constrained by the limited gate number and the presence of barren plateaus. In this work, we propose and numerically demonstrate a novel approach for VQAs, utilizing randomised quantum circuits to generate the variational wavefunction. We parameterize the distribution function of these random circuits using artificial neural networks and optimize it to find the solution. This random-circuit approach presents a trade-off between the expressive power of the variational wavefunction and time cost, in terms of the sampling cost of quantum circuits. Given a fixed gate number, we can systematically increase the expressive power by extending the quantum-computing time. With a sufficiently large permissible time cost, the variational wavefunction can approximate any quantum state with arbitrary accuracy. Furthermore, we establish explicit relationships between expressive power, time cost, and gate number for variational quantum eigensolvers. These results highlight the promising potential of the random-circuit approach in achieving a high expressive power in quantum computing. 
\end{abstract}

\maketitle

\section{Introduction}

In quantum computing, a central issue is to solve problems using as few quantum gates as possible. The reason is that the accuracy of quantum gates is inevitably degraded by noise, causing the possibility of errors in each gate operation. Consequently, when the gate number increases, quantum computing becomes more and more unreliable, ultimately resulting in catastrophic failure. Therefore, the number of gates must be minimised for realising any quantum computing application. This issue is particularly relevant for noisy intermediate-scale quantum (NISQ) computers, in which large-scale quantum error correction is unavailable \cite{lau2022nisq}. 

Over the past decade, tremendous effort has been put forward to develop quantum applications while taking the gate number as the primary consideration. Among such applications, variational quantum algorithms (VQAs) \cite{cerezo2021variational} have attracted the most interest. The key idea behind VQAs is constructing an ansatz quantum circuit $U(\theta)$ with parameterised quantum gates. Gate parameters $\theta$ are optimised in a feedback loop between quantum and classical computing. See the schematic diagram in Fig. \ref{fig:VQA}(a). The circuit generates a variational wavefunction $\ket{\phi(\theta)} = U(\theta)\ket{0}^{\otimes n}$. Ideally, once the optimal parameters are found, $\ket{\phi(\theta)}$ is the answer to the given problem. In this way, VQAs can solve certain problems keeping the gate count minimal. Within this variational framework, a diverse range of applications has been developed, such as the variational quantum eigensolver (VQE) \cite{peruzzo2014variational,kandala2017hardware,mcclean2017hybrid,higgott2019variational,nakanishi2019subspace,parrish2019quantum,huggins2020non,tilly2022variational}, variational quantum simulators \cite{li2017efficient,mcardle2019variational,yuan2019theory,endo2020variational,ollitrault2023quantum}, quantum approximate optimisation algorithm \cite{farhi2014quantum,wang2018quantum,zhou2020quantum,li2020quantum,wauters2020reinforcement,sack2021quantum}, quantum neural networks \cite{biamonte2017quantum,farhi2018classification,cong2019quantum,beer2020training,abbas2021power}, etc~\cite{mitarai2018quantum,havlivcek2019supervised,schuld2019quantum,perez2020data,schuld2020circuit,xu2021variational}. Notably, VQAs are well suited for today's NISQ technologies, leading to rapid progress in their experimental implementations \cite{cerezo2021variational,o2016scalable,kandala2017hardware,cirstoiu2020variational,kim2023evidence}. 

\begin{figure}[htbp]
\centering
\includegraphics[width=1\linewidth]{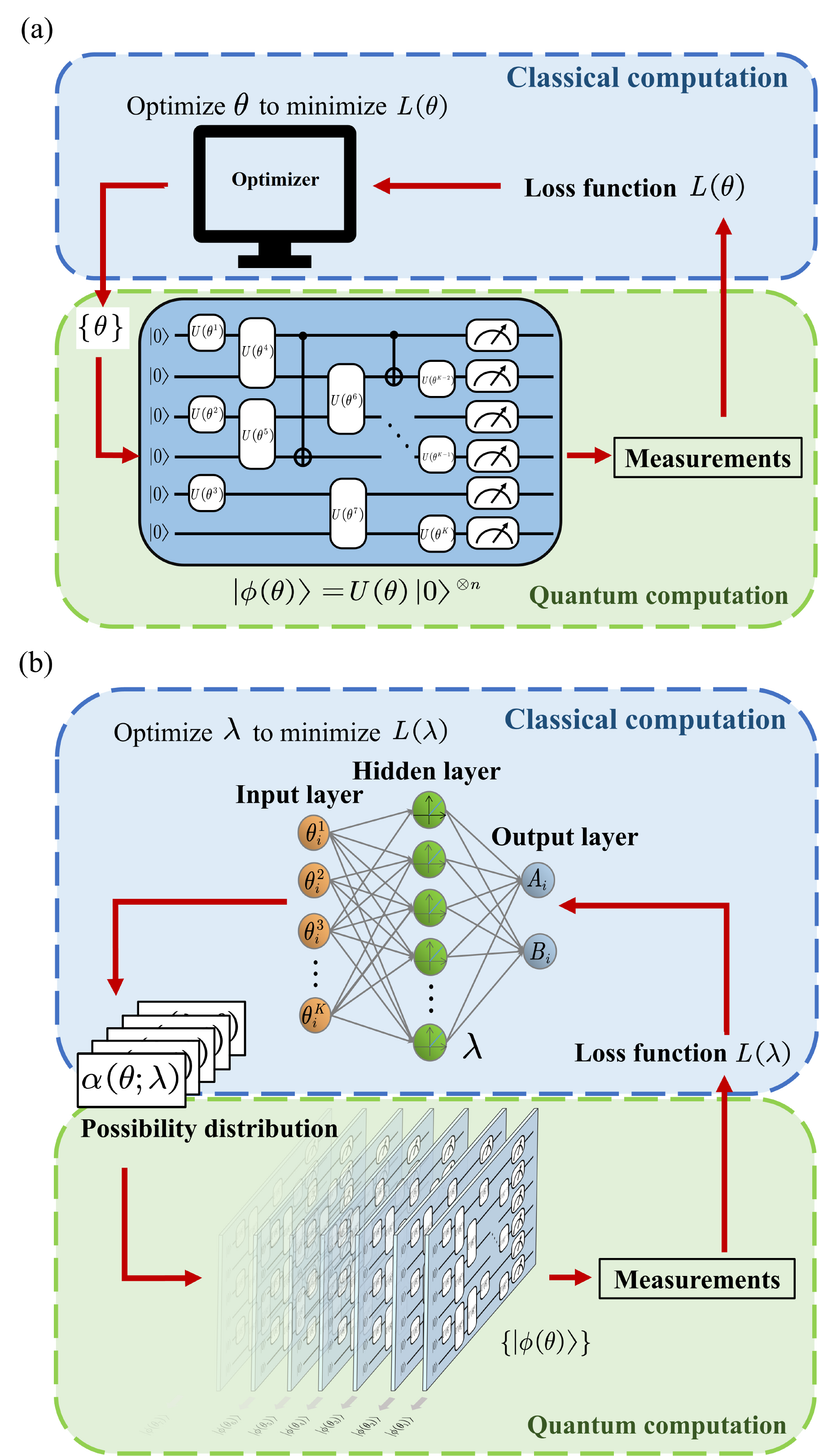}
\caption{(a) The schematic diagram of variational quantum algorithms using deterministic circuits. The problem is solved by finding the optimal circuit parameters $\theta$ to minimize loss function $L\left( \theta\right)$.
(b) The schematic diagram of variational quantum algorithms using random circuits. Circuits are sampled according to a parameterised guiding function $\alpha\left(\theta; \lambda\right)$, which determines the possibility distribution. Instead of $\theta$, we optimise the parameterised guiding function $\alpha\left(\theta; \lambda\right)$ to solve a problem. See Section~\ref{sec:VQCMC} for details.}
\label{fig:VQA}
\end{figure}

The success of VQAs depends on the expressive power of the variational wavefunction~\cite{sim2019expressibility,tangpanitanon2020expressibility,du2020expressive,nakaji2021expressibility,PhysRevLett.128.080506}. Since the variational wavefunction can only express a subset of all possible quantum states, a fundamental assumption is that the target quantum state is within this subset. Therefore, it is desirable to use a variational wavefunction with higher expressive power, capable of exploring a larger portion of the entire state space. By employing such a variational wavefunction, VQA has a better chance of finding the solution or attaining higher accuracy, as depicted in Fig. \ref{fig:expressiveness}. To maximise the expressive power, we confront two challenges. First, regarding the practical implementation, the gate number limits the circuit size and the number of parameters. Second, highly expressive circuits usually resemble the unitary $t$-design, which can lead to difficulties during the training process due to the problem of vanishing gradients \cite{mcclean2018barren,marrero2021entanglement,cerezo2021cost,PRXQuantum.3.010313}. Despite the challenges, much attention has been devoted to designing circuits \cite{PhysRevX.11.041011,volkoff2021large,ostaszewski2021structure,choquette2021quantum,PhysRevResearch.3.033090,PhysRevA.104.042418,PRXQuantum.2.020310,PhysRevX.12.011047,harrow2023approximate}. However, it remains largely unexplored that the expressive power can be improved by optimising the strategy of using circuits in the quantum-classical feedback loop. 

\begin{figure}[htbp]
\centering
\includegraphics[width=0.75\linewidth]{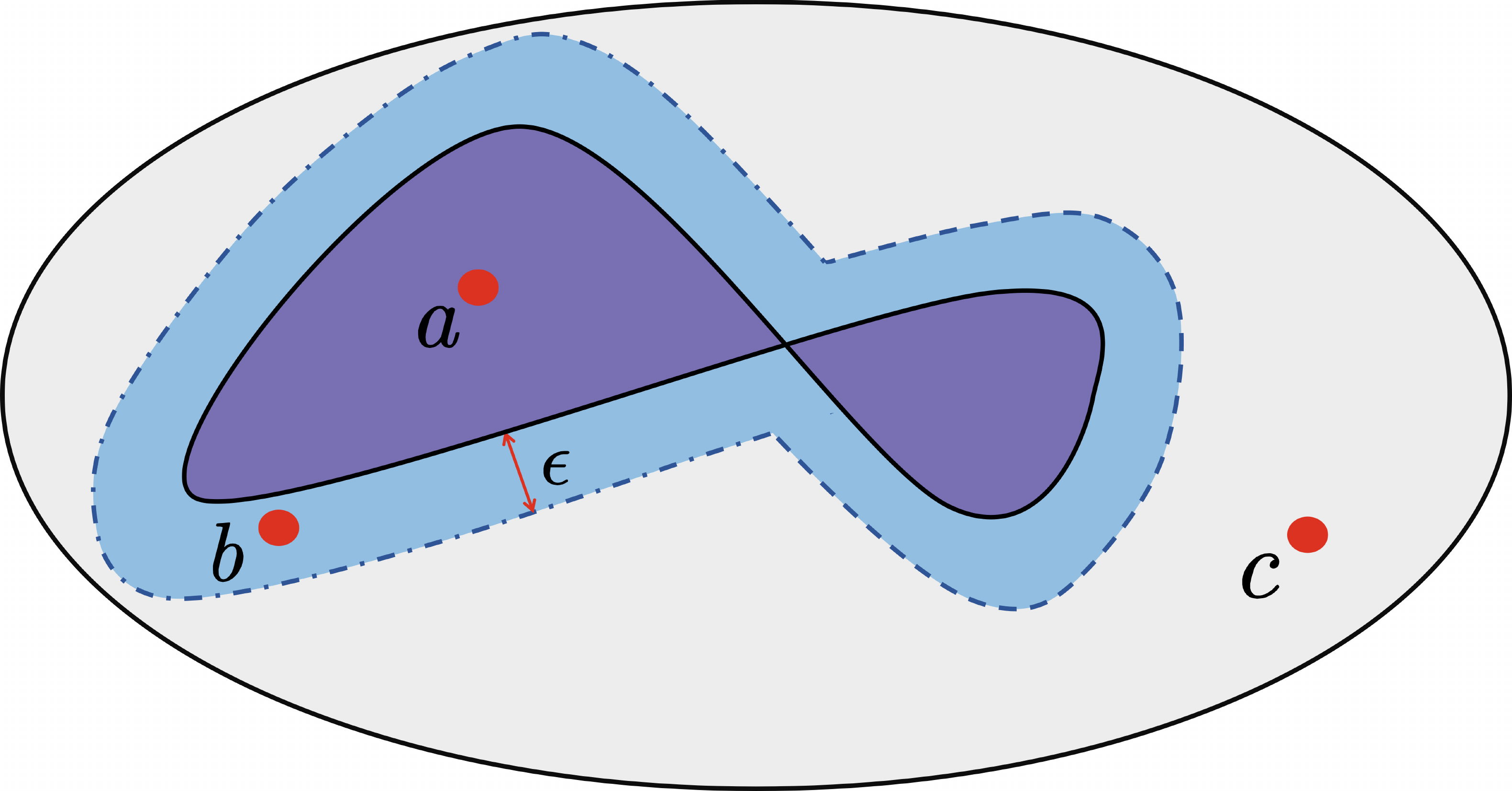}
\caption{
The subset of variational wavefunctions. In variational quantum algorithms (VQAs), the solution to a problem is represented by a quantum state. If the solution state is in the subset (e.g. state $a$) or close to the subset (e.g. state $b$) with a tolerable error $\epsilon$, VQA can successfully solve the problem (up to a proper optimiser and other practical issues). If the solution state is far from the subset (e.g. state $c$), VQA fails. 
}
\label{fig:expressiveness}
\end{figure}

In this work, we develop a new paradigm of utilising quantum circuits in VQAs, within which we can systematically increase the expressive power while the gate number stays the same. Instead of a deterministic quantum circuit, we consider a variational wavefunction that is the average of wavefunctions generated by random circuits, in the form $\ket{\psi(\lambda)} = \E[e^{i\gamma(\theta)}\ket{\phi(\theta)}]$. In contrast to deterministic-circuit VQAs, the distribution and phase functions are parameterised and optimised. Specifically, we propose to parameterise the distribution and phase with artificial neural networks (ANNs) \cite{jain1996artificial}. In the spirit of VQAs, we solve problems by optimising the parameterised distribution and phase functions, as illustrated in Fig. \ref{fig:VQA}(b). 

The key feature of the random-circuit approach is the trade-off between the expressive power and the time cost. When the gate number is fixed, we can increase the expressive power with an enlarged time cost. This feature is because of the statistical error caused by introducing randomness. To suppress the statistical error, we must repeat the measurement more times than in the deterministic-circuit approach, resulting in an enlarged time cost. We propose methods to control this time cost. When we apply a strong constraint on the time cost, the distribution is close to a delta function. In this case, the time cost and expressive power of $\ket{\psi(\lambda)}$ are almost the same as $\ket{\phi(\theta)}$. When we apply a weak constraint on the time cost, distributions with more randomness are allowed. Then, the time cost is larger, and the expressive power becomes higher. 

In the first part of this paper, we describe a general framework of VQAs using random circuits. Besides the framework, we also analyse the statistical error and introduce two methods for controlling the time cost. In the second part, we present a simple numerical demonstration taking the ground-state problem (i.e. the VQE algorithm) as an example. The numerical result illustrates the trade-off between expressive power and time cost. In the third part, we give a set of theorems to justify the performance and potential of the random-circuit approach. 

The rigorous theoretical results are listed as follows. Firstly, in the low-cost limit, the expressive power of $\ket{\psi(\lambda)}$ is not lower than $\ket{\phi(\theta)}$. This result is obtained using a rectifier ANN \cite{bengio2009learning,glorot2011deep,maas2013rectifier} to parameterise the distribution function. Secondly, to justify the power-cost trade-off, we show that increasing the time cost always enlarges the expressive power. Thirdly, in the high-cost limit, we prove the {\it universal approximation theorem} of the random-circuit approach: For many finite-size ansatz circuits, we can approximate an arbitrary state to arbitrary accuracy as long as the time cost is sufficiently large. Fourthly, to have a concrete discussion about this trade-off feature, we specifically focus on the ground-state problem and take the Hamiltonian ansatz circuit \cite{wecker2015progress} as an example. We obtain an upper bound of the error in approximating the ground state. We find that the error upper bound is a monotonic decreasing function of the time cost. Subject to a proper initial state, the error upper bound always vanishes in the limit of large time cost. Lastly, although VQAs are mainly developed for NISQ devices, we show the long-term potential of the random-circuit approach by analysing its performance with an increasing gate number. We prove that given a finite time cost, the error upper bound decreases with the gate number in the circuit. The gate number scales with the permissible error $\epsilon$ as $O(1/\epsilon^2)$. 

\section{Variational quantum-circuit Monte Carlo}\label{sec:VQCMC}

In this section, firstly, we introduce concepts necessary for understanding the random-circuit approach. Secondly, we give a general formalism of the approach. Finally, we analyse the statistical error and relate it to a quantity that can be evaluated with a quantum computer, and we introduce the methods for controlling the time cost. 

\subsection{Preliminaries}

Our approach is called variational quantum-circuit Monte Carlo (VQCMC), which is inspired by quantum Monte Carlo in classical computing \cite{foulkes2001quantum,austin2012quantum}. Quantum Monte Carlo methods combined with quantum computing have recently proposed, including auxiliary-field Monte Carlo~\cite{huggins2022unbiasing,xu2023quantum}, Green’s function
Monte Carlo~\cite{mazzola2022exponential,xu2023quantum}, variational Monte Carlo~\cite{montanaro2023accelerating,moss2023enhancing,xu2023quantum}, full configuration interaction Monte Carlo~\cite{zhang2022quantum,kanno2023quantum} and stochastic series expansion Monte Carlo~\cite{tan2022sign}. VQCMC consists of two main components: \textit{sample space} and \textit{guiding function}. In this section, we introduce the two components and present a way of parameterising the guiding function using an ANN. 


\subsubsection{Sample space}

The sample space $\Omega$ is a set of quantum states from which we draw random samples. A quantum computer can prepare states in the form $\ket{\phi(\theta)} = U(\theta)\ket{0}^{\otimes n}$, where $\ket{0}^{\otimes n}$ is the initial state of $n$ qubits, $U(\theta)$ is the unitary operator of a parameterised quantum circuit, and $\theta$ is a vector of parameters. An example of the parameterised circuit is illustrated in Fig.~\ref{fig:VQA}(b). For those familiar with conventional deterministic-circuit VQAs, $U(\theta)$ could be one of the prominent ansatz circuits, e.g. unitary coupled cluster ansatz \cite{peruzzo2014variational}, hardware-efficient ansatz \cite{kandala2017hardware}, Hamiltonian ansatz \cite{wecker2015progress}, ADAPT-VQE \cite{grimsley2019adaptive}, etc. All these states generated by the circuit form the sample space $\Omega = \{\ket{\phi(\theta)}\st \theta\in\Theta\}$, where $\Theta$ denotes the space of parameters ($\Theta \subseteq \mathbb{R}^K$ when the circuit has $K$ real parameters).

\subsubsection{Guiding function}

The guiding function $\alpha(\theta;\lambda)$ determines the distribution of states in sample space. It also determines the phase factor $e^{i\gamma}$. In this work, we parameterise the guiding function as ANN, as shown in Fig. \ref{fig:VQA}(b). ANNs, consisting of interconnected neurons, has the ability to approximate a function with an arbitrary accuracy~\cite{cybenko1989approximation, hornik1989multilayer, hornik1991approximation}. Recently, they have been used to represent wavefunctions in classical computing \cite{carleo2017solving,luo2019backflow,schutt2019unifying,hermann2020deep,pfau2020ab,ren2023towards}. In our case, we use an ANN to parameterise the guiding function $\alpha(\theta;\lambda)$. The input to the network is the circuit parameter vector $\theta$, and the output is a complex number $\alpha$.
The map from input to output depends on network parameters, represented by $\lambda$. Details of ANN are given in Sec. \ref{sec:RNN}. We solve problems by optimising $\lambda$. 

The distribution of states $\ket{\phi(\theta)}$ is described by a probability density function that is proportional to the guiding function, i.e. 
\begin{eqnarray}
P(\theta;\lambda) &=& \frac{\abs{\alpha(\theta;\lambda)}}{C(\lambda)},
\end{eqnarray}
where 
\begin{eqnarray}
C(\lambda) &=& \int d\theta \abs{\alpha(\theta;\lambda)}.
\end{eqnarray}
is the normalisation factor. Note that $\theta$ is the random variable, and $\lambda$ represents the parameters of the distribution. Given the guiding function, we can sample $\theta$ from the distribution using Markov chain Monte Carlo methods, e.g. the Metropolis-Hastings algorithm~\cite{hastings1970MC}. And the phase factor is 
\begin{eqnarray}
e^{i\gamma(\theta;\lambda)} &=& \frac{\alpha(\theta;\lambda)}{\abs{\alpha(\theta;\lambda)}}.
\end{eqnarray}

\subsubsection{Rectifier artificial neural network}
\label{sec:RNN}

\begin{figure}[htbp]
\centering
\includegraphics[width=\linewidth]{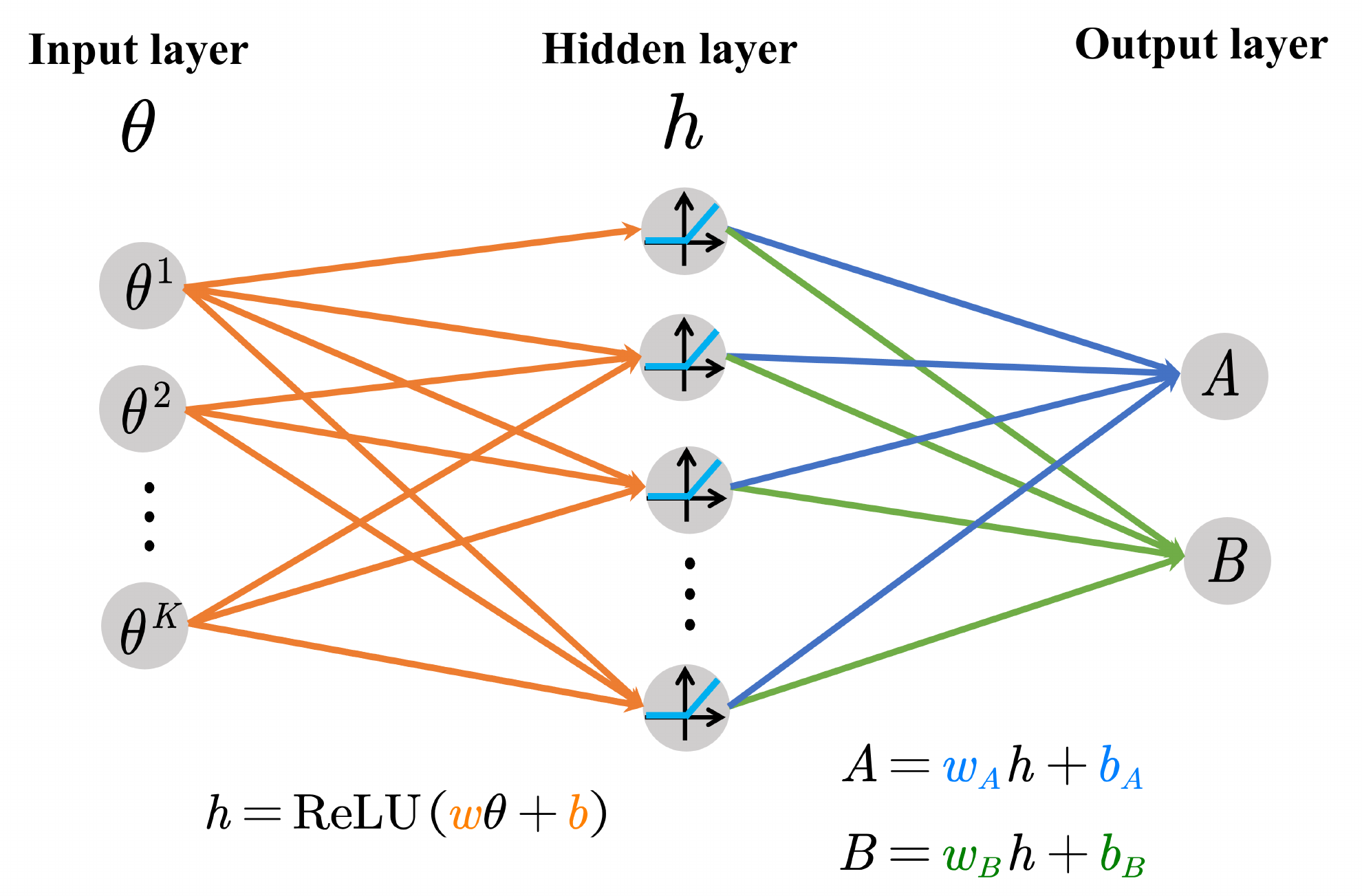}
 \caption{
The rectifier neural network with one hidden layer for parameterising the guiding function $\alpha(\theta; \lambda)$, where $\lambda=\left( w,b,w_A,w_B,b_A,b_B \right)$. 
}
\label{fig:ANN}
\end{figure}

We take the rectifier ANN as an example. Rectifier ANN is a feed-forward neural network with the rectifier activation function $\Relu(z)=\max\{0,z\}$~\cite{brownlee2019gentle}. The activation function is an essential part of ANNs, making them have nonlinear expression ability. Activation functions include the rectifier activation function, sigmoid activation function, tanh activation function, etc. Rectifier feed-forward neural networks is one of the most generally used activation functions. Rectifier feed-forward neural networks, one of the most commonly used activation functions, stand out as a potent tool in various applications due to their exceptional attributes such as rapid convergence rate, swift learning speed, and streamlined expression for simplified calculations.

In the theoretical analysis, we take the rectifier ANN with only one 
hidden layer to parameterise the guiding function, see Fig.~\ref{fig:ANN}. It is straightforward to generalise to the case of multiple hidden layers and other activation functions. In the numerical experiment, we use a neural network with two hidden layers to have better training results. Without loss of generality, we suppose that the circuit parameter $\theta = (\vartheta_1,\vartheta_2,\ldots,\vartheta_K)^\T$ is an $K$-dimensional real vector (For clarity, we use the curly theta to denote the component of $\theta$). Then the input layer has $K$ neurons. Suppose the hidden layer has $L$ neurons. The output of the hidden layer is an $L$-dimensional real vector $h = \Relu(w\theta + b)$. The output layer has only two neurons corresponding to the magnitude and phase of the guiding function $\alpha$, respectively. Their outputs are two real numbers $A = w_Ah+b_A$ and $B = w_Bh+b_B$. The guiding function is parameterised as 
\begin{eqnarray}
\alpha(\theta; \lambda) = W(A)e^{iB}F(\theta),
\end{eqnarray}
where $\lambda = (w,b,w_A,b_A,w_B,b_B)$ is the parameter vector of the network, $W(A)$ is a non-negative function [for instance, $W(A)=\abs{A}$, $W(A)=e^{-A}$ and $W(A) = e^{-\frac{1}{2}A^2}$], and $F(\theta)$ is a prior guiding function (a complex-valued function in general). 

We introduce the prior guiding function to incorporate knowledge/intuition about the target state, such as which states $\ket{\phi(\theta)}$ may contribute more to the target state. We can also use the prior guiding function to confine $\theta$ to a subset of $\Theta$, i.e. taking $F(\theta)=0$ for $\theta$ outside the subset. Without any intuition or confinement, we can take $F(\theta) = 1$. 


\subsection{General formalism}
\label{sec:formalism}

The core idea of VQCMC is expressing a quantum state as a weighted average of states in the sample space. Given the expression, we can evaluate the expected value of operators with the Monte Carlo method and then the loss function used to find optimal parameters. 

In VQAs, we solve problems by optimising a variational wavefunction. In the deterministic-circuit approach, if $U(\theta)$ is the ansatz circuit, the state $\ket{\phi(\theta)}$ is the variational wavefunction, i.e. we aim to find the solution in the state subset $\Omega$. 

In the random-circuit approach, we aim to express the quantum state as an average of states in $\Omega$. Specifically, the variational wavefunction reads 
\begin{eqnarray}
\ket{\psi(\lambda)} = \int d\theta P(\theta;\lambda)e^{i\gamma(\theta;\lambda)} \ket{\phi(\theta)}.
\label{eq:psi}
\end{eqnarray}
Here, both the distribution $P(\theta;\lambda)$ and the phase $\gamma(\theta;\lambda)$ are determined by the guiding function $\alpha(\theta;\lambda)$. Notice that this variational wavefunction is a function of network parameters $\lambda$. We optimise $\lambda$ to solve problems, in contrast with optimising circuit parameters $\theta$ in the deterministic-circuit approach. 

In many VQAs, we determine the parameters of the variational wavefunction by minimising a loss function. In this work, we focus on VQAs in this category, including VQE, quantum approximate optimisation algorithm and quantum neural networks. However, we would like to note that the random-circuit approach can also be employed in variational quantum simulators that update parameters following certain differential equations. 

We take VQE as an example, in which the loss function is the expected value of a Hamiltonian operator $H$. Because the variational wavefunction $\ket{\psi(\lambda)}$ is unnormalised, the expected value of $H$ reads 
\begin{eqnarray}
L(\lambda) = \frac{\bra{\psi(\lambda)}H\ket{\psi(\lambda)}}{\braket{\psi(\lambda)}{\psi(\lambda)}}.
\end{eqnarray}
By minimising the loss function, we can find the optimal $\ket{\psi(\lambda)}$ to approximate the ground state of $H$. 

In general, the loss function is a function of such expected values. We define the expected value of an operator $O$ in the state $\ket{\psi(\lambda)}$ as 
\begin{eqnarray}
\mean{O}(\lambda) = \bra{\psi(\lambda)}O\ket{\psi(\lambda)},
\end{eqnarray}
and the expected value in the {\it normalised} state as 
\begin{eqnarray}
\widetilde{\mean{O}}(\lambda) = \frac{\bra{\psi(\lambda)}O\ket{\psi(\lambda)}}{\braket{\psi(\lambda)}{\psi(\lambda)}} = \frac{\mean{O}(\lambda)}{\mean{\openone}(\lambda)},
\end{eqnarray}
where $\openone$ is the identity operator. With the notations, the loss function in VQE can be rewritten as $L(\lambda) = \widetilde{\mean{H}}(\lambda)$. Let $O_1,O_2,\ldots$ be a set of operators. 
 The general loss function reads 
\begin{eqnarray}
L(\lambda) = L\Big(\widetilde{\mean{O_1}}(\lambda),\widetilde{\mean{O_2}}(\lambda),\ldots\Big).
\end{eqnarray}
We need to evaluate such loss functions to carry out VQAs. 

Next, we present the method of evaluating an expected value $\mean{O}(\lambda)$ by sampling states in $\Omega$, and we also give a detailed pseudocode. 

\subsubsection{The operator estimator}

We only demonstrate how to evaluate $\mean{O}(\lambda)$. Then, it is straightforward to evaluate expected values in the normalised state $\widetilde{\mean{O}}(\lambda)$ and the loss function $L(\lambda)$. 

Given the expression of the variational wavefunction in Eq. (\ref{eq:psi}), we can rewrite the expected value as 
\begin{eqnarray}
\mean{O}(\lambda) &=& \int d\theta d\theta' P(\theta;\lambda)P(\theta';\lambda) \notag \\
&&\times e^{i[\gamma(\theta; \lambda)-\gamma(\theta'; \lambda)]} \bra{\phi(\theta')}O\ket{\phi(\theta)}.
\label{eq:meanO}
\end{eqnarray}
In this way, we have transformed the expected value in a quantum state into the expected value of the quantity 
\begin{eqnarray}
X_{\theta,\theta'}(\lambda) = e^{i[\gamma(\theta; \lambda)-\gamma(\theta'; \lambda)]} \bra{\phi(\theta')}O\ket{\phi(\theta)}.
\end{eqnarray}
According to the Monte Carlo method, we can evaluate $\mean{O}(\lambda)$ by computing the sample mean in randomly drawn $(\theta,\theta')$. 

We measure the quantity $X_{\theta,\theta'}$ on a quantum computer. 
We can use the Hadamard test~\cite{ekert2002direct} to measure such a quantity, including methods with and without ancilla qubit. See Appendix \ref{app:HT} for a brief review of these methods. Regardless of which method we use, quantum computing outputs an estimate of $X_{\theta,\theta'}$, denoted by $\hat{X}_{\theta,\theta'}$. Due to the randomness in quantum measurement, $\hat{X}_{\theta,\theta'}$ is usually inexact. In what follows, we suppose that $\hat{X}_{\theta,\theta'}$ is unbiased and has a finite variance, which is true in the Hadamard test. 

The pseudocode for evaluating $\mean{O}(\lambda)$ is shown in Algorithm \ref{alg:estimator}. $\hat{\mean{O}}$ is the output value. 

\begin{figure}
\begin{minipage}{\linewidth}
\begin{algorithm}[H]
{\small
\begin{algorithmic}[1]
\caption{{\small Operator estimator.}}
\label{alg:estimator}
\Function{OpeEstimator}{$\lambda$,O,$M$}
\For{$l=1$ to $M$}
\State Draw $(\theta,\theta')$ according to the distribution $P(\theta;\lambda)P(\theta';\lambda)$. 
\State $\hat{X}_{\theta,\theta'} \gets$ \Call{QuantumComputing}{$\lambda$,$\theta$,$\theta'$,$O$}
\State $\mu_l \gets \hat{X}_{\theta,\theta'}$
\EndFor
\State Output $\hat{\mean{O}} \gets \frac{1}{M}\sum_{l=1}^{M} \mu_l$. 
\EndFunction
\end{algorithmic}
}
\end{algorithm}
\end{minipage}
\end{figure}

\subsection{Statistical errors}

There are two sources of statistical errors: quantum measurement and randomised circuits. We will show that both of them can be suppressed by taking a large sample size $M$ (see Algorithm \ref{alg:estimator}). We will also show that statistical errors are amplified when $\mean{\openone}(\lambda)$ is small. Therefore, we have to take a larger sample size (i.e. time cost) when $\mean{\openone}(\lambda)$ is smaller. Through $\mean{\openone}(\lambda)$, we can control the time cost. 

For the first error source, when we use a quantum circuit to measure the quantity $X_{\theta,\theta'}$, the measurement outcome in each run of the circuit is random. To reduce this statistical error due to quantum measurement, we can repeat the measurement. For example, in the Hadamard test, we take $\hat{X}_{\theta,\theta'}$ as the mean of outcomes in the repeated measurements. Then, the variance of $\hat{X}_{\theta,\theta'}$ decreases with the number of measurements, denoted by $M_Q$. In the ancilla-qubit Hadamard test, suppose $O$ is a unitary Hermitian operator, the variance of $\hat{X}_{\theta,\theta'}$ has the upper bound 
\begin{eqnarray}
\Var\left(\hat{X}_{\theta,\theta'}\right) \leq \sigma_O^2 = \frac{1}{M_Q}.
\end{eqnarray}
The result is similar for a general operator: for a general operator, $\sigma_O^2$ depends on properties of $O$ and details of the measurement protocol, see Appendix \ref{app:HT}. 

For the second error source, $\mean{O}$ is the mean of the random variable $X_{\theta,\theta'}$, which has fluctuation. 

Overall, the variance of $\hat{\mean{O}}$ has two terms and is proportional to $1/M$: 
\begin{eqnarray}
\Var\left(\hat{\mean{O}}\right) = \frac{1}{M}\left(A+B-\mean{O}^2\right),
\end{eqnarray}
where 
\begin{eqnarray}
A &=& \int d\theta d\theta' P(\theta;\lambda)P(\theta';\lambda) \Var\left(\hat{X}_{\theta,\theta'}\right), \\
B &=& \int d\theta d\theta' P(\theta;\lambda)P(\theta';\lambda) \abs{X_{\theta,\theta'}}^2.
\end{eqnarray}
Here, $A$ and $B$ correspond to the first and second error sources, respectively.

Let $\norm{O}$ be the spectral norm of $O$. Then, $B\leq \norm{O}^2$, and the overall variance has the upper bound 
\begin{eqnarray}\label{eq:variance}
\Var\left(\hat{\mean{O}}\right) \leq \frac{1}{M}\left(\norm{O}^2+\sigma_O^2\right).
\end{eqnarray}
When $\sigma_O^2\propto 1/M_Q$, it is optimal to take $M_Q = 1$ to minimise the total number of measurements $MM_Q$ required for achieving a certain value of the variance. 

To complete the discussion on statistical properties of $\hat{\mean{O}}$, we note that when $\hat{X}_{\theta,\theta'}$ is unbiased, $\hat{\mean{O}}$ is also unbiased. 

\subsubsection{Sign problem}

In classical computing, quantum Monte Carlo suffers from the sign problem. It arises when the probability distribution sampled in the Monte Carlo simulation has a complex phase, resulting in a cancellation of positive and negative contributions. VQCMC has a similar problem: when the phase of $X_{\theta,\theta'}$ is oscillatory, the absolute value of $\mean{O}$ is small; then the relative error is large. The large relative error can lead to a large absolute error when evaluating expected values in the normalised state $\widetilde{\mean{O}}$. 

Let $\delta_{O} = \hat{\mean{O}}-\mean{O}$ be the error in $\hat{\mean{O}}$. The error in evaluating $\widetilde{\mean{O}}$ is 
\begin{eqnarray}
\frac{\hat{\mean{O}}}{\hat{\mean{\openone}}} - \widetilde{\mean{O}} \simeq \frac{\delta_{O}-\widetilde{\mean{O}}\delta_{\openone}}{\mean{\openone}}.
\end{eqnarray}
We can find that the error is amplified by the factor of $1/\mean{\openone}$. To suppress the error, we have to take a sufficiently large $M$, such that $\delta_{O}$ and $\delta_{\openone}$ are sufficiently small compared with $\mean{\openone}$. Therefore, when $\mean{\openone}$ is small, the sample size is large. This observation is summarised in the following theorem, and the proof is in Appendix~\ref{app:sign}. 

\begin{theorem}\label{the:sign}
Let $\varepsilon$ and $\kappa$ be any positive numbers. When the sample size satisfies 
\begin{eqnarray}
M\geq \frac{\chi^2}{\kappa\varepsilon^2\mean{\openone}^2}, 
\label{eq:Mbound}
\end{eqnarray}
where 
\begin{eqnarray}
\chi = \sqrt{\norm{O}^2+\sigma_O^2}+\sqrt{1+\sigma_{\openone}^2}\norm{O}+\varepsilon\sqrt{1+\sigma_{\openone}^2},
\end{eqnarray}
the statistical error is smaller than $\varepsilon$ with the probability 
\begin{eqnarray}
\Pr\left(\left\vert\frac{\hat{\mean{O}}}{\hat{\mean{\openone}}} - \widetilde{\mean{O}}\right\vert\leq \varepsilon\right)\geq 1-2\kappa.
\end{eqnarray}
\end{theorem}

\subsubsection{Modified loss functions and time-cost control}

According to Theorem~\ref{the:sign}, we can control the time cost (i.e. the sample size) by confining the variational wavefunction in the region where $\mean{\openone}$ is large. Here, we propose two methods for this purpose. The first method is generic for all VQAs that work through minimising a loss function. The second method is specific to VQE. We will give a rigorous theoretical justification of the second method. 

In the first method, we modify the loss function and take 
\begin{eqnarray}
L'(\lambda) = L(\lambda)-z\tanh\left[\frac{\mean{\openone}(\lambda)-x}{y}\right],
\label{eq:Lprime}
\end{eqnarray}
where the $\tanh$ function is added to the raw loss function. The $\tanh$ function plays the role of a barrier at $\mean{\openone}(\lambda)=x$ (see Fig. \ref{fig:confinement}): In the vicinity of the barrier, the loss function increases rapidly ($y$ determines how rapidly) when $\mean{\openone}(\lambda)$ decreases. The height of the barrier is determined by $z$. With the $\tanh$ function, $\lambda$ prefers to stay where $\mean{\openone}>x$, thus controlling the time cost. Some other functions, e.g. the rectifier and sigmoid functions, can be used in a similar way (After a preliminary numerical test, we find that the $\tanh$ may perform better than the other two). 

\begin{figure}[htbp]
\centering
\includegraphics[width=1\linewidth]{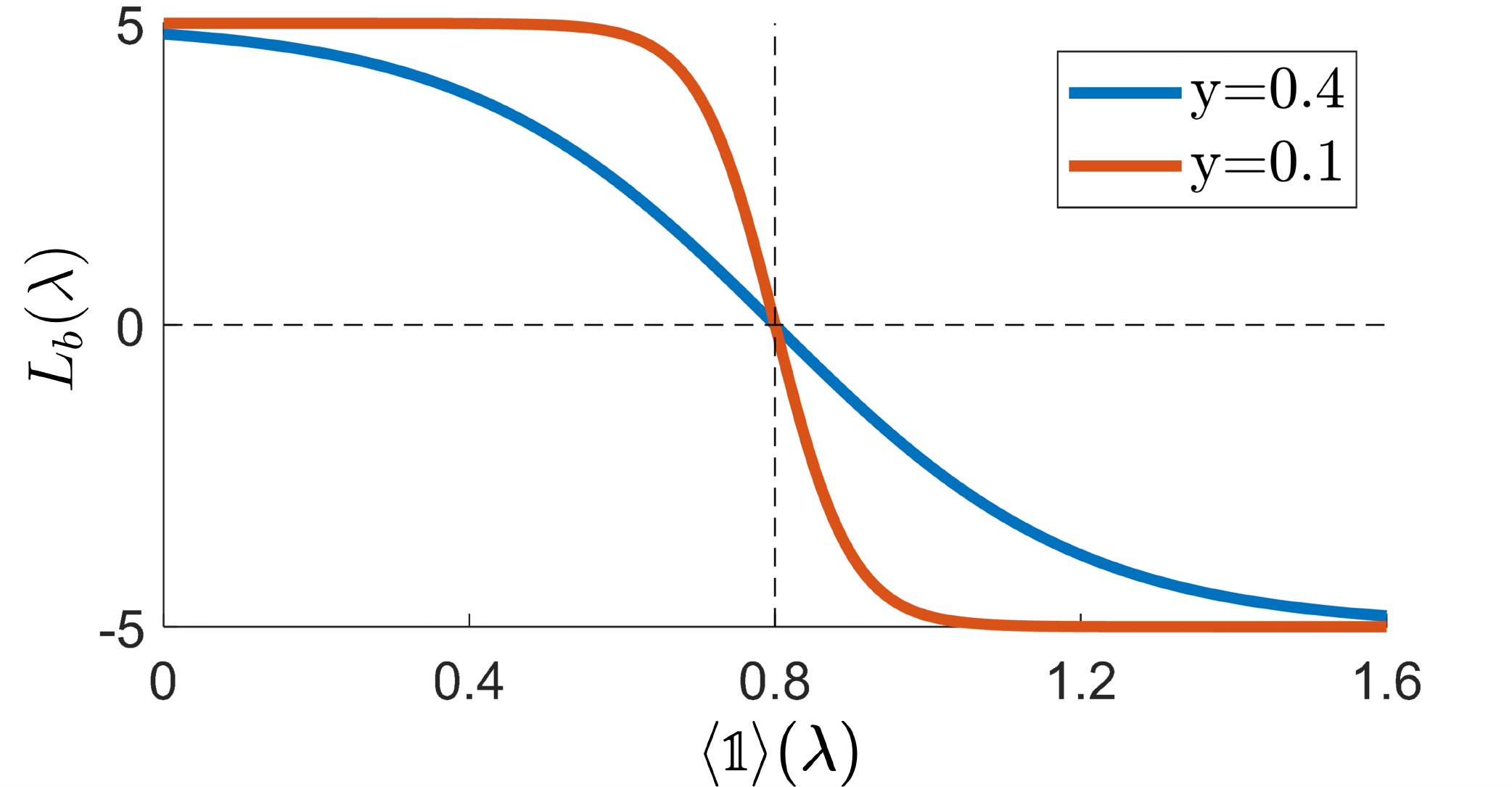}
\caption{
The barrier term $L_b(\lambda) = -z\tanh\left[\frac{\mean{\openone}(\lambda)-x}{y}\right]$ in the modified loss function, where the parameter $x$ specifies the barrier's position. We illustrate this with two examples: for $y = 0.4$ and $y = 0.1$, while holding $z = 5$ and $x = 0.8$ constant.
}
\label{fig:confinement}
\end{figure}

In the second method, we regularise the loss function of VQE by adding positive constants to the denominator and numerator, i.e. 
\begin{eqnarray}
L''(\lambda) = \frac{\mean{H}(\lambda)+\eta_H}{\mean{\openone}(\lambda)+\eta_{\openone}}.
\end{eqnarray}
Now we explain why such a loss function can prevent $\mean{\openone}$ from vanishing. The exact value of $\mean{\openone}$ is always positive. Adding $\eta_{\openone}$ to the denominator makes sure that the denominator is positive (with a certain probability) even with the presence of statistical error, noticing that $\hat{\mean{\openone}}$ may be negative due to the statistical error. When the denominator is positive, adding $\eta_H$ to the numerator increases the value of the loss function. The increment is larger when the denominator is smaller. Therefore, minimising the modified loss function can prevent the denominator from vanishing. 

By taking proper values of $\eta_{\openone}$ and $\eta_H$, the modified loss function $L''(\lambda)$ has two good properties. First, it is variational, i.e. $L''(\lambda)\geq E_g$, where $E_g$ is the ground-state energy of $H$. Therefore, minimising $L''(\lambda)$ always leads to a better result of the ground-state energy. The estimator of the modified loss function inherits this property. Second, we have an analytical upper bound of the statistical error in the energy (The raw loss function $L(\lambda) = \widetilde{\mean{H}}(\lambda)$ is the expected value of the energy in the normalised state). These two properties as summarised in the following theorem, and the proof is in Appendix~\ref{app:variational}. 

\begin{theorem}\label{the:variational}
Let $\kappa$ be any positive number. Take 
$$
\eta_H = \sqrt{\frac{\norm{H}^2+\sigma_H^2}{\kappa M}}\text{ and }
\eta_{\openone} = \sqrt{\frac{1+\sigma_{\openone}^2}{\kappa M}}.
$$
An estimate of the modified loss function 
\begin{eqnarray}
\hat{L}''(\lambda) = \frac{\hat{\mean{H}}(\lambda)+\eta_H}{\hat{\mean{\openone}}(\lambda)+\eta_{\openone}}
\end{eqnarray}
is in the interval 
\begin{eqnarray}
\min\{\widetilde{\mean{H}}(\lambda),0\} \leq \hat{L}'' \leq \widetilde{\mean{H}}(\lambda) + 2\frac{\eta_H+\norm{H}\eta_{\openone}}{\mean{\openone}(\lambda)}
\label{eq:inequality}
\end{eqnarray}
with a probability of at least $1-2\kappa$. 
\end{theorem}

We can find that the estimator of the modified loss function is variational when the true ground-state energy $E_g$ is negative. The negativity condition can always be satisfied by subtracting a sufficiently large positive constant from the Hamiltonian, i.e.~taking $H\leftarrow H-E_{const}$. When $E_g$ is negative, $\hat{L}''\geq E_g$ holds in the entire parameter space up to a controllable probability (Notice that $\widetilde{\mean{H}}(\lambda)\geq E_g$). 

\section{Numerical demonstration}

In this section, we demonstrate the random-circuit approach by numerical simulations. We take VQE as an example of VQAs and solve the ground state of an anti-ferromagnetic Heisenberg model. We choose the barrier method to control the time cost, and we illustrate the trade-off between expressive power and time cost. 

\begin{figure}[htbp]
\centering
\includegraphics[width=0.7\linewidth]{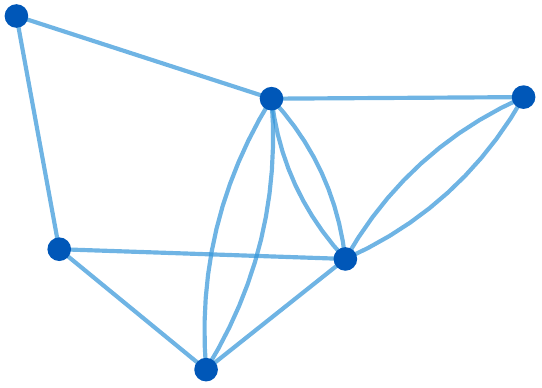}
\caption{Graph of the Heisenberg model.}
\label{fig:graph}
\end{figure}

The anti-ferromagnetic Heisenberg model is on a randomly generated graph, as shown in Fig.~\ref{fig:graph}. On the graph, each vertex represents a spin-1/2 particle, and each edge represents the interaction between the two spins. The Hamiltonian reads
\begin{eqnarray}
H = J\sum_{\langle i,j\rangle} (X_iX_j+Y_iY_j+Z_iZ_j),
\end{eqnarray}
where $X_i$, $Y_i$ and $Z_i$ denote Pauli operators of the $i$-th qubit, $J$ is the coupling strength, and $\langle i,j\rangle$ denotes two connected spins on the graph. We take the coupling strength $J$ such that the Hamiltonian is normalised by the spectral norm, i.e. $\norm{H} = 1$. 


We parameterise the circuit with a simplified version of the Hamiltonian ansatz. The circuit has only two parameters $\theta_{XY}$ and $\theta_Z$: 
\begin{eqnarray}
U(\theta) = \left[R(\theta)\right]^{N_T} R_0,
\end{eqnarray}
where 
\begin{eqnarray}
R(\theta) = \prod_{\langle i,j\rangle} e^{-iZ_iZ_j\theta_Z/2}e^{-iY_iY_j\theta_{XY}/2}e^{-iX_iX_j\theta_{XY}/2},
\end{eqnarray}
$N_T$ denotes the number of gate layers, $R_0$ prepares the pairwise singlet state 
\begin{eqnarray}
\ket{\Psi_0} &=& R_0\ket{0}^{\otimes n} = \ket{\Phi}_{1,2}\otimes \ket{\Phi}_{3,4}\otimes \cdots,
\end{eqnarray}
and 
\begin{eqnarray}
\ket{\Phi}_{i,j} = \frac{1}{\sqrt{2}}\left(\ket{0}_i\otimes\ket{1}_j-\ket{1}_i\otimes\ket{0}_j\right)
\end{eqnarray}
is the singlet state of spins $i$ and $j$. 

\begin{figure}[htbp]
\centering
\includegraphics[width=\linewidth]{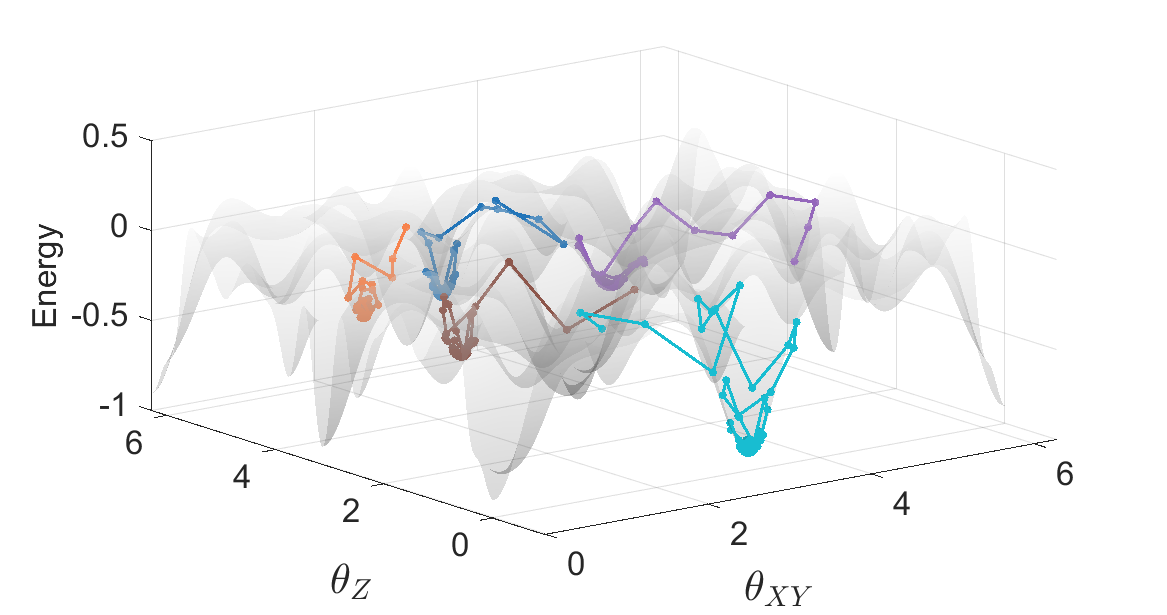}
\caption{Energy landscape of the simplified Hamiltonian ansatz circuit with $N_T = 2$. The colored lines represent the training process of finding the minimum value using the VQA method from different initial points.}
\label{fig:surface}
\end{figure}

With this simplified Hamiltonian ansatz circuit, we can plot the energy landscape as shown in Fig. \ref{fig:surface}. This energy landscape has undesired properties. First, it has a rugged surface. When using gradient descent and randomly choosing the initial point, we find that the point easily falls into a local minimum. Second, even for the global minimum, its error in the energy is about $0.05$ (The exact ground-state energy is $-1$, and the global minimum is about $-0.95$). In other words, the expressive power of the ansatz circuit is insufficient for approximating the ground state to achieve an error smaller than $0.05$. Therefore, the performance of the ansatz circuit is poor in the deterministic-circuit approach. 

In the random-circuit approach, we can solve the ground state to a satisfactory accuracy even with the poor-performance ansatz circuit. By controlling the time cost, we can systematically increase the expressive power and approaches the exact solution with a sufficiently large time cost. 

In the random-circuit approach, we parameterise the guiding function using a rectifier ANN with two hidden layers, each consisting of $200$ neurons. We take $\theta = (\theta_{XY}, \theta_Z)$ as the input to ANN. For the output function, we take $W(A) =e^{-A}$. We choose the prior guiding function $F(\theta) = \sum_{i=1}^{N_\theta} \delta(\theta-\theta_i)$, where $\theta_i$ are $N_\theta = 100$ uniformly generated points in the parameter space. With this prior guiding function, we effectively utilise a finite sample space to simplify the numerical simulation.

\begin{figure}[htbp]
\centering
   \includegraphics[width=\linewidth]{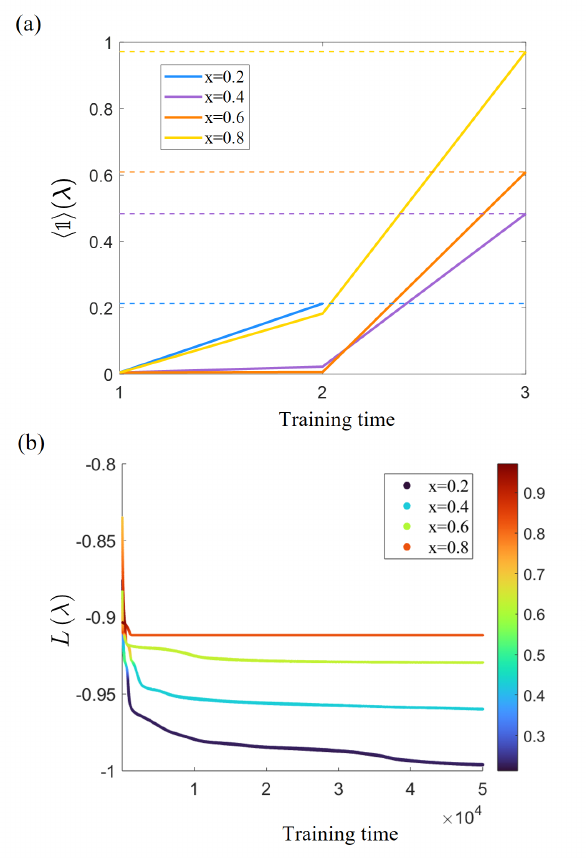}
\caption{
The two-stage training in the barrier method. (a) In the first stage, we minimise the loss function $L_b$ and it ends when $\mean{\openone}(\lambda)>x$ is satisfied. Here, we take $x = 0.2$, $0.4$, $0.8$ as examples.   (b) In the second stage, we take the resulting value of $\lambda$ trained in the first stage as the initial value, then we minimise the loss function $L'(\lambda)$. $L(\lambda)$ decreases with the training times. The colorbar represents the value of $\mean{\openone}(\lambda)$. In both stages, we take $z = 5$,  $y=1$, and we use a gradient descent algorithm Adam to minimise loss functions. 
}
\label{fig:trail}
\end{figure}

Implementing the barrier method requires a proper initial value of the parameter $\lambda$. Suppose the initial $\lambda$ is randomly chosen in the parameter space, the corresponding value of $\mean{\openone}(\lambda)$ may be small and violate the $\mean{\openone}(\lambda)>x$ restriction. In this case, the statistical error can be large and cause problems in the initial stage of training. To avoid this issue, we have to choose the initial value satisfying $\mean{\openone}(\lambda)>x$. We give two methods of doing this. First, we can take $\lambda$ such that $P(\theta; \lambda)\approx\delta(\theta-\theta_0)$ is a delta-function distribution. For such a distribution $\mean{\openone}(\lambda)\approx 1$. In Sec. \ref{sec:superiority}, we will show how to take $\lambda$ to approximate the delta-function distribution by the ANN. Second, we can employ two-stage training, as shown in Fig. \ref{fig:trail}. In the first stage, we can take $L_b(\lambda) = -z\tanh\left[\frac{\mean{\openone}(\lambda)-x}{y}\right]$ as the loss function to find a value of $\lambda$ satisfying $\mean{\openone}(\lambda)>x$. 
This $\tanh$ loss function is robust to the statistical error in $\mean{\openone}(\lambda)$ even when $\mean{\openone}(\lambda)$ is small [in contrast to $L(\lambda)$, in which $\mean{\openone}(\lambda)$ is the denominator]. In the second stage, we take the value $\lambda$ determined in the first stage as the initial value and minimise $L'(\lambda)$. We obtain our numerical results with the two-stage method. 

\begin{figure}[htbp]
\centering
\includegraphics[width=\linewidth]{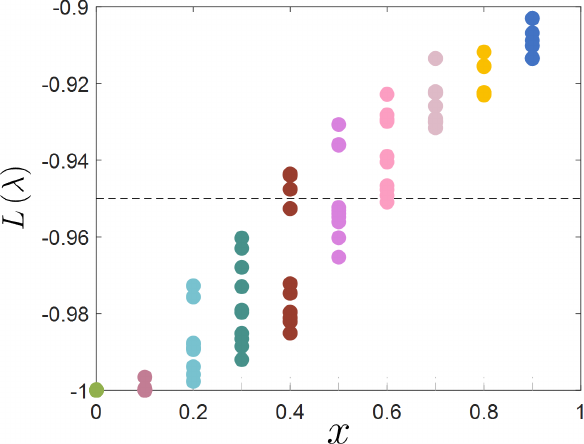}
\caption{The ground-state energy estimated in the random-circuit approach with different barriers. Each point represents the final value of $L(\lambda)$ in a numerical experiment. For each value of $x$, the experiment is repeated for $10$ times. The dashed line indicates the global minimum energy in the deterministic-circuit approach.
}
\label{fig:tanhdenominator1}
\end{figure}

We set the barrier at $x = 0.1,0.2,\ldots,0.9$. Then, the denominator $\mean{\openone}(\lambda)$ is confined in the $\mean{\openone}(\lambda)>x$ regime, respectively. For each value of $x$, we repeat the numerical experiment ten times. The results are shown in Fig. \ref{fig:tanhdenominator1}. We can find that the error decreases when we reduce the value of $x$, which illustrates the power-cost trade-off. 

\section{Theoretical results on the expressive power}

In this section, we present a few theoretical results on the expressive power of the random-circuit approach. First of all, let us introduce the ways of characterising {\it expressive power} and {\it time cost}. 

We characterise the expressive power in two straightforward ways: the subset of variational wavefunctions and the error in the ground-state energy. Discussions based on the subset are general for various VQAs tasks. When focusing on VQE, we employ the energy error. 

We define the variational-wavefunction subset of the $\ket{\psi(\lambda)}$ as 
\begin{eqnarray}
\calV = \left\{\left.\ket{\widetilde{\psi}(\lambda)}=\frac{\ket{\psi(\lambda)}}{\sqrt{\mean{\openone}(\lambda)}} \,\right\vert \lambda\in\Lambda\right\},
\end{eqnarray}
where $\Lambda$ denotes the space of parameters [$\lambda = (w,b,w_A,b_A,w_B,b_B)$ has $N = KL+3L+2$ real parameters; $w\in\mathbb{R}^{L\times K}$, $b,w_A,w_B\in\mathbb{R}^L$ and $b_A,b_B\in\mathbb{R}$; therefore, $\Lambda = \mathbb{R}^N$]. Notice that states in this subset are normalised. The normalisation factor $\mean{\openone}(\lambda)$ is related to the sample size $M$ (see Theorem \ref{the:sign}). Given a finite permitted run time, the sample size is finite, and we can only utilise a subset of $\calV$ in VQAs: We cannot evaluate a state to an adequate accuracy when $\mean{\openone}(\lambda)$ is too small. To reflect this time cost constraint, we define 
\begin{eqnarray}
\calV_x = \left\{\left.\ket{\widetilde{\psi}(\lambda)} \,\right\vert \mean{\openone}(\lambda)\geq x\right\}.
\end{eqnarray}
For a larger subset $\calV_x$, its expressive power is higher. As shown in Appendix \ref{app:covering}, there is a simple relation between the variational-wavefunction subset and the covering number of the hypothesis space (a measure of expressive power). 

For the variational-wavefunction subset $\calV_x$, we define the minimum error in the ground-state energy as 
\begin{eqnarray}
\epsilon_g(\calV_x) = \min_{\ket{\widetilde{\psi}(\lambda)}\in \calV_x} \bra{\widetilde{\psi}(\lambda)}H\ket{\widetilde{\psi}(\lambda)} - E_g.
\end{eqnarray}
The aim of VQE is to find the ground state of a Hamiltonian. The error $\epsilon_g$ describes how well the variational wavefunction can approximate the ground state. 

We characterise the time cost with $x$. According to Theorem~\ref{the:sign}, the sample size $M$ is proportional to $1/\mean{\openone}(\lambda)^2$. When a larger $M$ is taken, we can search for the solution in a subset $\calV_x$ with a smaller $x$. Specifically, given the sample size $M$, an upper bound of the statistical error $\varepsilon$ and failure probability $2\kappa$, we can search for the solution in $\calV_x$ with $x = \chi/(\varepsilon\sqrt{\kappa M})$. In short, when the time cost is larger, $x$ is smaller, and vice versa. 

\subsection{Low-cost limit}
\label{sec:superiority}

The time cost of evaluating a state in the random-circuit approach is determined by the normalisation factor $\mean{\openone}(\lambda)$. The largest value that $\mean{\openone}(\lambda)$ can take is one, i.e. $\mean{\openone}(\lambda) = 1$ corresponds to the low-cost limit. 

Our first theoretical result is that $\Omega$ is a subset of the closure of $\calV_x$ (for all $x<1$). In other words, we can express all states $\ket{\phi(\theta)}\in \Omega$ to arbitrary accuracy with the random-circuit variational wavefunction $\ket{\psi(\lambda)}$, and the normalisation factor $\mean{\openone}(\lambda)$ approaches one. Therefore, the expressive power of $\ket{\psi(\lambda)}$ in the low-cost limit is not lower than $\ket{\phi(\theta)}$. 

We assume that $\ket{\phi(\theta)}$ has a finite gradient with respect to $\theta$. Let $\mathbf{n}$ be a unit vector in the space of $\theta$. We assume there exists a positive number $\xi$ such that for all $\mathbf{n}$ and $\theta$, 
\begin{eqnarray}\label{eq:gradient}
\norm{\mathbf{n}\cdot\nabla\ket{\phi(\theta)}}_2 \leq \xi.
\end{eqnarray}
Here, $\norm{\bullet}_2$ denotes the $\ell^2$ norm. This assumption is true for all ansatz circuits to the best of our knowledge. For example, if each parameter component $\vartheta_i$ is the angle of a rotation gate $e^{-i\sigma_i\vartheta_i}$, $\norm{\frac{\partial}{\partial \vartheta_i} \ket{\phi(\theta)}}_2 = 1$; Then, $\norm{\mathbf{n}\cdot\nabla\ket{\phi(\theta)}}_2 \leq \sqrt{K}$ ($K$ is the dimension of $\theta$). Under this assumption, we have the following theorem. 

\begin{theorem}
Suppose the guiding function $\alpha$ is parameterised as a rectifier ANN, the hidden layer has $L \geq 2K$ neurons, $W(A) = e^{-A}$, and $F(\theta) = 1$. Then, for all $\ket{\phi(\theta)}\in\Omega$ and $x<1$, 
\begin{eqnarray}
\inf_{\ket{\widetilde{\psi}(\lambda)}\in\calV_x} \norm{\ket{\widetilde{\psi}(\lambda)}-\ket{\phi(\theta)}}_2 = 0.
\label{eq:superiority}
\end{eqnarray}
\label{the:superiority}
\end{theorem}

The proof is given in Appendix~\ref{app:superiority}. To approximate $\ket{\phi(\theta)}$ with the random-circuit variational wavefunction, we consider values of ANN parameters as follows [see Fig. \ref{fig:ANN_param}(a)], we use $\theta'$ to denote the input to the neural network for clarity): i) For the $j$th hidden neuron, we take $w_{j,i} = \delta_{j,2i-1}-\delta_{j,2i}$ and $b_j = -\sum_i(\delta_{j,2i-1}-\delta_{j,2i})\vartheta_i$; ii) For the output neuron $A$, we take $w_{A,j} = \tau$ and $b_A = K\ln (2/\tau)$; and iii) For the output neuron $B$, we take $w_{B,j} = b_B = 0$. With these parameters, hidden neurons realise the modulus calculation, i.e.~$x_{2i-1}+x_{2i} = \abs{\vartheta'_i-\vartheta_i}$. Outputs of the neural network are $A = \tau\sum_i\abs{\vartheta'_i-\vartheta_i} + K\ln (2/\tau)$ and $B = 0$. Here, $\tau$ is a positive number.  Then, the corresponding random-circuit variational wavefunction is 
\begin{eqnarray}
\ket{\psi(\lambda)} = \left(\frac{\tau}{2}\right)^K \int d\theta' e^{-\tau\norm{\theta'-\theta}_1} \ket{\phi(\theta')}.
\label{eq:DCpsi}
\end{eqnarray}
In the limit $\tau\rightarrow +\infty$, $\ket{\psi(\lambda)}$ [as well as the normalised state $\ket{\widetilde{\psi}(\lambda)}$] approaches $\ket{\phi(\theta)}$, and the distance between the two states vanishes. 

We remark that when $\ket{\psi(\lambda)}$ approaches $\ket{\phi(\theta)}$, the corresponding value of $\mean{\openone}(\lambda)$ approaches one. Therefore, we can approximate $\ket{\phi(\theta)}$ with $\ket{\psi(\lambda)}$ in the low-cost limit. 

\subsection{Power-cost trade-off}

Our second theoretical result is a simple observation that justifies the trade-off between expressive power and time cost. The following proposition is straightforward, and we present it without proof. It shows that the variational-wavefunction subset $\calV_x$ enlarges monotonically when $x$ decreases. Therefore, the expressive power increases monotonically with the time cost. 

\begin{proposition}
For all $x_1\leq x_2$, $\calV_{x_1} \supseteq \calV_{x_2}$. 
\end{proposition}

\subsection{High-cost limit: Universal approximation theorem}

Now, we discuss the extreme case of expressive power in the random-circuit approach. We ask whether all states in the Hilbert space can be approximated with the random-circuit variational wavefunction when an arbitrarily large time cost is allowed. 

Suppose $\Omega$ is a spanning set of the entire Hilbert space, all states can be expressed in the form 
\begin{eqnarray}
\ket{\varphi} = \int d\theta \beta(\theta)\ket{\phi(\theta)}.
\label{eq:varphi}
\end{eqnarray}
If we can approximate $\beta(\theta)$ with the guiding function $\alpha(\theta;\lambda)$, we can approximate $\varphi$ with the variational wavefunction $\ket{\psi(\lambda)}$. Owing to the universal approximation theorem of ANNs, we can approximate an arbitrary continuous function to arbitrary accuracy~\cite{csaji2001approximation}. Therefore, we have the following theorem. The proof is in Appendix \ref{app:UAT}.

\begin{theorem}
Suppose the guiding function $\alpha$ is parameterised as a rectifier neural network, the hidden layer has $L$ neurons, $W(A) = e^{-A}$, and $F(\theta) = 1$. Under conditions i) $\Omega$ is a spanning set of the Hilbert space $\mathcal{H}$, ii) the circuit parameter space $\Theta$ is compact, and iii) $\ket{\phi(\theta)}$ has a finite gradient with respect to $\theta$, for all normalised states $\ket{\varphi}\in\mathcal{H}$ and $\epsilon>0$, there exist $L\in\mathbb{N}$ and $\lambda\in\Lambda$ such that 
\begin{equation}
\norm{\ket{\widetilde{\psi}(\lambda)}-\ket{\varphi}}_2 < \epsilon.
\end{equation}
\label{the:UAT}
\end{theorem}

The compact condition holds when each component is a rotation angle of a single-qubit gate, i.e. its value is in the interval $[0,2\pi]$. The spanning-set condition holds for many ansatz circuits, for example, a circuit with one layer of single-qubit rotation gates on the initial state $\ket{0}^{\otimes n}$. 

\subsection{Hamiltonian ansatz: Energy error versus time cost}

In what follows, we discuss the trade-off feature taking the ground-state problem and the Hamiltonian ansatz as an example. We take the error in the ground-state energy $\epsilon_g(\calV_x)$ as the measure of the expressive power. We show that 1) an upper bound of $\epsilon_g(\calV_x)$ decreases with $1/x$; 2) the error upper bound approaches zero in the limit $1/x\rightarrow \infty$. It is noteworthy that the error upper bound approaches zero for Hamiltonian ansatz with any circuit depth. 

\subsubsection{Hamiltonian ansatz}

The Hamiltonian ansatz is a parameterised Trotterisation circuit~\cite{wecker2015progress}. Suppose that $H = \sum_{j=1}^{N_H}h_j\sigma_j$, where $\sigma_j$ are Pauli operators, $h_j$ are real coefficients, and $N_H$ is the number of terms. The Hamiltonian ansatz reads 
\begin{eqnarray}
U(\theta) = R(\omega_{N_T})\cdots R(\omega_2)R(\omega_1)R_0,\label{eq:vha1}
\label{eq:Utheta}
\end{eqnarray}
where $R_0$ is a unitary operator that prepares the initial state $\ket{\Psi_0} = R_0\ket{0}^{\otimes n}$, 
\begin{eqnarray}
R(\omega_i) = e^{-i\sigma_{N_H}\omega_{i,N_H}}\cdots e^{-i\sigma_{2}\omega_{i,2}} e^{-i\sigma_{1}\omega_{i,1}},\label{eq:vha2}
\label{eq:Romega}
\end{eqnarray}
is the unitary operator of one Trotter step, $\omega_{i,j}$ is the angle of the $j$th rotation gate in the $i$th Trotter step, and $N_T$ is the number of Trotter steps. We take the circuit parameter vector as $\theta = (t,\omega_1,\omega_2,\ldots,\omega_{N_T})$, in which $t$ is a redundant real parameter. The usage of the redundant parameter will be shown later. 

\subsubsection{ANN configuration and prior guiding function}

We use a rectifier ANN with $L \geq 2$ hidden-layer neurons to parameterise the guiding function, and we take $W(A) = e^{-\frac{1}{2}A^2}$. We choose a non-trivial prior guiding function $F(\theta)$ inspired by the Pauli-operator-expansion formula of real-time evolution~\cite{suzuki1990fractal}; See Appendix \ref{app:PGF}. The prior guiding function has the following property: 
\begin{eqnarray}
\int d\omega_1\cdots d\omega_{N_T} F(\theta) R(\omega_{N_T})\cdots R(\omega_1) = e^{-iHN_Tt}.
\label{eq:PGF}
\end{eqnarray}

\begin{figure}[htbp]
\centering
    \includegraphics[width=\linewidth]{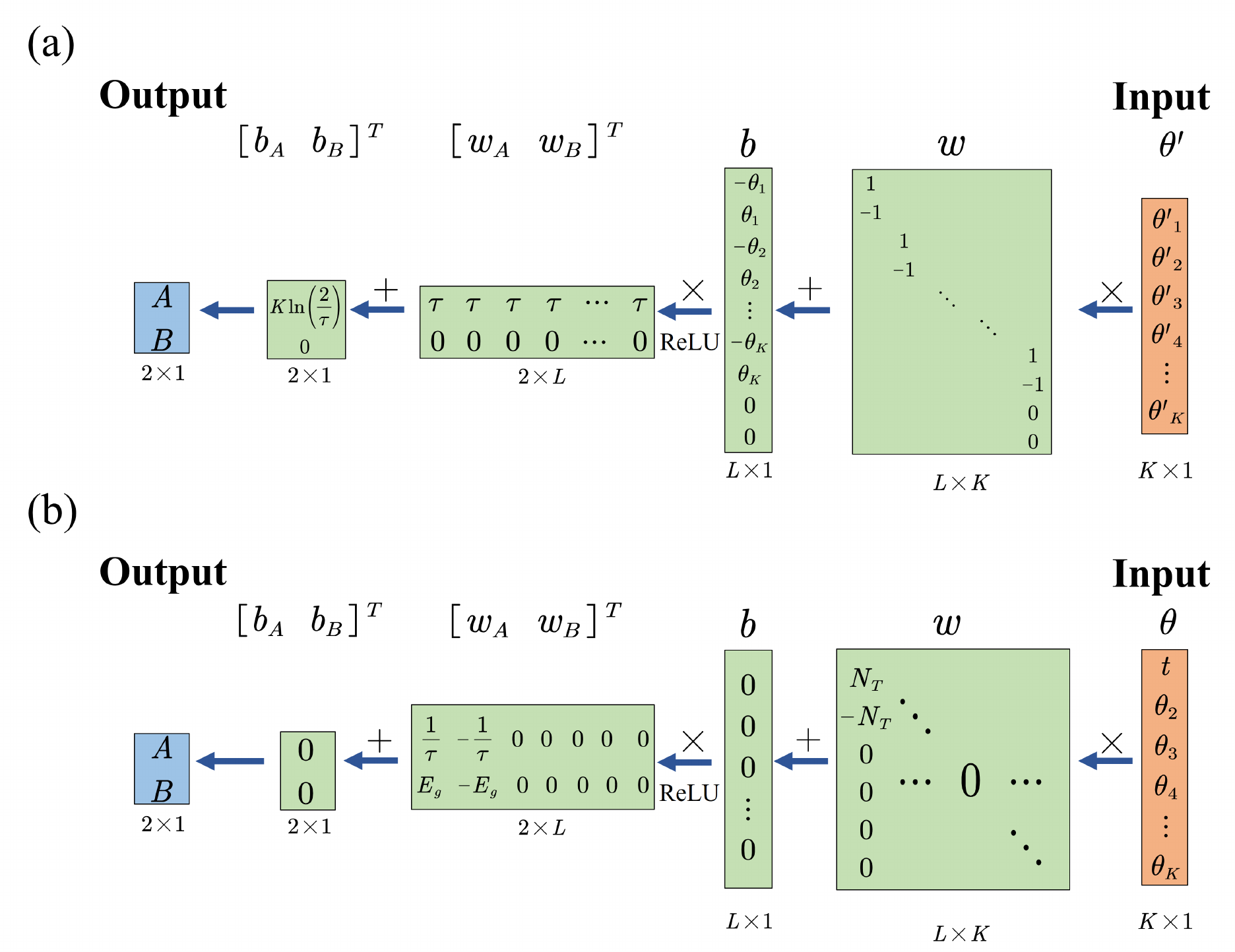}
\caption{
(a) The values of ANN parameters when approximating $\ket{\phi(\theta)}$ with the random-circuit variational wavefunction. 
(b) The values of ANN parameters when considering Hamiltonian Ansatz to approximate the ground state.
}
\label{fig:ANN_param}
\end{figure}
To approximate the ground state, we consider values of ANN parameters as follows [see Fig. \ref{fig:ANN_param}(b)]: i) We suppose that $\vartheta_1 = t$ is the first input neuron; ii) For the $j$th hidden neuron, we take $w_{j,i} = (\delta_{j,1}-\delta_{j,2})\delta_{i,1}N_T$ and $b_j = 0$, such that only the first two hidden neurons are non-trivial; iii) For the output neuron $A$, we take $w_{A,j} = (\delta_{j,1}-\delta_{j,2})\tau^{-1}$ and $b_A = 0$; and iv) For the output neuron $B$, we take $w_{B,j} = (\delta_{j,1}-\delta_{j,2})E_g$ and $ b_B = 0$. Here, $\tau$ is a positive number.

With the above configuration of parameters, outputs are $A = \tau^{-1}N_Tt$ and $B = E_gN_Tt$. Then, the corresponding variational wavefunction is 
\begin{eqnarray}
\ket{\psi(\lambda)} &=& \frac{1}{C(\lambda)}\int d\theta e^{-\frac{N_T^2t^2}{2\tau^2}}e^{iE_gN_Tt}F(\theta) U(\theta)\ket{0}^{\otimes n} \notag \\
&=& \frac{\tau\sqrt{2\pi}}{C(\lambda)N_T} e^{-\frac{1}{2}(H-E_g)^2\tau^2}\ket{\Psi_0}
\label{eq:HApsi}
\end{eqnarray}
The operator $e^{-\frac{1}{2}(H-E_g)^2\tau^2}$ is a result of the integral over $t$~\cite{wall1995extraction}, and it projects the initial state onto the ground state in the limit $\tau\rightarrow +\infty$. 

\subsubsection{Projection onto the ground state}
\label{sec:projection}

The operator $e^{-\frac{1}{2}(H-E_g)^2\tau^2}$ partially projects the initial state onto the ground state when $\tau$ is finite. Let $\Delta$ be the energy gap between the ground state $\ket{\Psi_g}$ and the first excited state. When $\Delta\geq (\sqrt{2}\tau)^{-1}$, the error in the energy decreases exponentially with $\tau$; otherwise, the error is inversely proportional to $\tau$. 

\begin{lemma}
Let $E(\lambda)$ be the energy of the state in Eq.~(\ref{eq:HApsi}). Then, the error in the energy has the upper bound 
\begin{eqnarray}
E(\lambda) - E_g \leq \frac{1-p_g}{p_g} e^{-\tilde{\Delta}^2\tau^2} \tilde{\Delta},
\label{eq:error}
\end{eqnarray}
where $p_g = \abs{\braket{\Psi_g}{\Psi_0}}^2$ and $\tilde{\Delta} = \max\{\Delta,(\sqrt{2}\tau)^{-1}\}$. 
\label{lem:projection}
\end{lemma}

The proof is given in Appendix~\ref{app:projection}. 

\subsubsection{Monotonic dependence on the time cost}

First, we give a lower bound of the normalisation factor $\mean{\openone}(\lambda)$, which determines the time cost. See the following lemma, and the proof is in Appendix~\ref{app:cost}. 

\begin{lemma}
Let $\mean{\openone}(\lambda)$ be the normalisation factor of the state in Eq.~(\ref{eq:HApsi}). Then, the normalisation factor has the lower bound 
\begin{eqnarray}
\mean{\openone}(\lambda) \geq \frac{2\pi p_g}{c(\tau)^2},
\label{eq:cost}
\end{eqnarray}
where 
\begin{eqnarray}
c(\tau) = \int du e^{-\frac{u^2}{2}} \left(2e^{\frac{h_{tot}\tau}{N_T}\abs{u}}-1-2\frac{h_{tot}\tau}{N_T}\abs{u}\right)^{N_T},
\label{eq:ctau1}
\end{eqnarray}
and $h_{tot} = \sum_{j=1}^{N_H}\abs{h_j}$. When $N_T>4h_{tot}^2\tau^2$, 
\begin{eqnarray}
c(\tau) &\leq & \sqrt{\frac{2\pi}{1-4h_{tot}^2\tau^2/N_T}}.
\label{eq:ctau2}
\end{eqnarray}
\label{lem:cost}
\end{lemma}

Then, with the lower bound in Eq. (\ref{eq:cost}), we can conclude that the state in Eq.~(\ref{eq:HApsi}) with $\tau = c^{-1}(\sqrt{2\pi p_g/x})$ is in the subset $\calV_x$. Here, we have used that $c(\tau)$ is strictly monotonic, and $c^{-1}(\bullet)$ denotes the inverse function. Then, we can apply the error upper bound in Eq. (\ref{eq:error}) to $\calV_x$. We have the following result (which holds for all $N_T$), and the proof is straightforward. 

\begin{theorem}
Take the Hamiltonian ansatz circuit. Suppose the guiding function $\alpha$ is parameterised as a rectifier neural network, the hidden layer has $L \geq 2$ neurons, and $W(A) = e^{-\frac{1}{2}A^2}$. With a proper prior guiding function $F(\theta)$, the error in the energy satisfies 
\begin{eqnarray}
\epsilon_g(\calV_x) \leq \frac{1-p_g}{p_g} e^{-\tilde{\Delta}^2\tau^2} \tilde{\Delta},
\label{eq:epsgub}
\end{eqnarray}
where $\tau = c^{-1}(\sqrt{2\pi p_g/x})$. The error upper bound decreases monotonically with $1/x$ and approaches zero ($\tau$ approaches $\infty$) in the limit $1/x\rightarrow \infty$. 
\label{the:monotone}
\end{theorem}

Notice that $1/x\propto \sqrt{M}$. Therefore, the error upper bound decreases monotonically with the time cost. We remark that a condition of the above result is $p_g>0$, which is a requirement on the initial state $\ket{\Psi_0}$. 

\subsection{Hamiltonian ansatz: Scaling with the circuit depth}

In Theorem~\ref{the:monotone}, the monotonic relation between the error upper bound and the time cost factor $1/x$ holds for all circuit depths $N_T$. It raises the question: What is the advantage of using a more powerful quantum computer that can realise quantum circuits with more gates? We answer the question in this section. We show that for all $x<p_g$, the error upper bound decreases with the circuit depth and vanishes in the limit of a large depth. 

The upper bound of $c(\tau)$ in Eq. (\ref{eq:ctau2}) implies a lower bound of $\tau = c^{-1}(\sqrt{2\pi p_g/x})$. When $x<p_g$, 
\begin{eqnarray}
c^{-1}(\sqrt{2\pi p_g/x}) \geq \frac{1}{2h_{tot}}\sqrt{N_T\left(1-\frac{x}{p_g}\right)}.
\label{eq:cinv}
\end{eqnarray}
Therefore, the upper bound of $\epsilon_g(\calV_x)$ in Eq. (\ref{eq:epsgub}) decreases with $N_T$. In the limit $N_T\rightarrow\infty$, the error upper bound approaches zero. 

Let $\epsilon$ be the permissible error in the ground-state energy. According to Lemma~\ref{lem:projection}, we can approximate the ground state and satisfy the permissible error $\epsilon$ by taking 
\begin{eqnarray}
\tau = \left\{\begin{array}{ll}
\frac{1-p_g}{\sqrt{2e}p_g\epsilon}, & \epsilon>\frac{1-p_g}{\sqrt{e}p_g}\Delta, \\
\frac{1}{\Delta}\sqrt{\ln\frac{(1-p_g)\Delta}{p_g\epsilon}}, & \epsilon\leq \frac{1-p_g}{\sqrt{e}p_g}\Delta.
\end{array}\right.
\label{eq:tau}
\end{eqnarray}
This value of $\tau$ determines the required circuit depth. Substituting $\tau$ into $\tau = c^{-1}(\sqrt{2\pi p_g/x})$ and taking into account Eq. (\ref{eq:cinv}), we can work out $N_T$. 

\begin{theorem}
Take the same setup as in Theorem \ref{the:monotone}. For all $\epsilon>0$ and all $x<p_g$, $\epsilon_g(\calV_x)\leq \epsilon$ holds when 
\begin{eqnarray}
N_T = \left\lceil\frac{4h_{tot}^2\tau^2}{1-x/p_g}\right\rceil = O\left(\frac{1}{\epsilon^2}\right).
\end{eqnarray}
Here, $\tau$ is given by Eq. (\ref{eq:tau}). 
\end{theorem}

Because the gate number is proportional to $N_T$, the gate number scales as $O(1/\epsilon^2)$.

\section{Conclusions}

In this work, we have presented a method of realising a variational wavefunction with randomised quantum circuits in VQAs. This method can systematically increase the expressive power without changing the gate number. The cost is an enlarged time for evaluating a quantity in the Monte Carlo calculation. The trade-off between the expressive power and time cost is analysed theoretically and illustrated numerically. Especially in VQE, we have shown that the energy error, that is due to the finite expressive power, decreases with the time cost and eventually vanishes. These results demonstrate the viability and potential advantages of the random-circuit approach in VQAs, providing a pathway for improving the performance of quantum computing with the constraints of gate numbers and associated errors. 


\begin{acknowledgments}
The authors acknowledge the helpful discussions with Yuxuan Du and Xiao Yuan.
This work is supported by the National Natural Science Foundation of China (Grant Nos. 12225507, 12088101 and 92265208) and NSAF (Grant No. U1930403 and U2330201).
\end{acknowledgments}

\appendix

\section{Hadamard test}
\label{app:HT}

In this section, we review two methods for measuring $X_{\theta,\theta'}$ on a quantum computer: Hadamard test with and without ancilla qubit. In the methods, we analyse the variance of $\hat{X}_{\theta,\theta'}$ in the ancilla-qubit Hadamard test in detail. 

\subsection{Ancilla-qubit Hadamard test}

\begin{figure}[htbp]
\centering
\includegraphics[width=\linewidth]{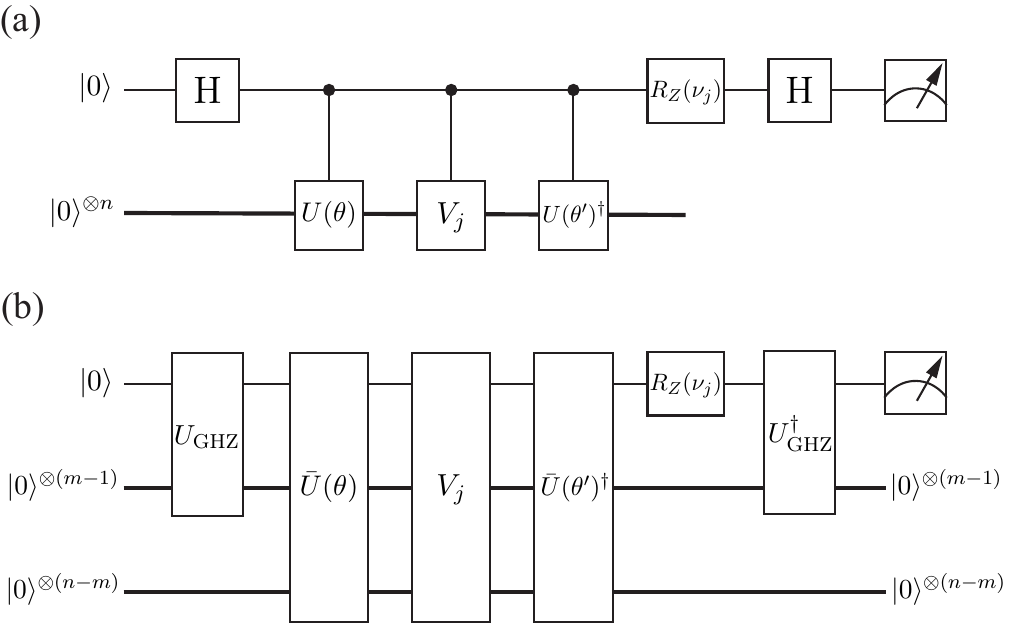}
\caption{(a) Hadamard test with an ancilla qubit. The rotation gate $R_Z(\nu_j) = e^{-iZ\nu_j/2}$. (b)Hadamard test without ancilla qubit. $U_{\rm GHZ} = \prod_{i=1}^{m-1} \Lambda_{i,i+1} {\rm H}_1$, where $\Lambda_{i,j}$ is a controlled-NOT gate with the $i$-th qubit as the control qubit and the $j$-th qubit as the target qubit, and ${\rm H}_1$ is a Hadamard gate on the first qubit. The gate $U_{\rm GHZ}$ prepares a GHZ state $\ket{\rm GHZ} = \frac{1}{\sqrt{2}}(\ket{0}^{\otimes m}+\ket{1}^{\otimes m})$.}
\label{fig:circuits}
\end{figure}

The conventional Hadamard test circuit requires an ancilla qubit~\cite{ekert2002direct}, as shown in Fig.~\ref{fig:circuits} (a) . In the final state of the circuit, the expected value of the ancilla-qubit $Z$ Pauli operator is 
\begin{eqnarray}
\mean{Z_A} = \Re \left[e^{i\nu_j} \bra{\phi(\theta')}V_j\ket{\phi(\theta)}\right].
\end{eqnarray}
Therefore, at the end of the circuit, we measure the ancilla qubit in the computational (i.e.~$Z$) basis. For each circuit shot, we obtain a measurement outcome $\mu = \pm 1$, and $\mean{Z_A} = \E[\mu]$. 

The measurement protocol depends on the property of operator $O$. The simplest case is that $O$ is unitary and Hermitian. In this case, we have
\begin{eqnarray}
X_{\theta,\theta'} = \Re \left[e^{i\nu} \bra{\phi(\theta')}O\ket{\phi(\theta)}\right],
\end{eqnarray}
where $\nu = \gamma(\lambda,\theta)-\gamma(\lambda,\theta')$. This quantity can be measured by taking $V_j=O$ and $\nu_j=\nu$. If the number of circuit shots is $M_Q$, the variance of $\hat{X}_{\theta,\theta'}$ has the upper bound $\sigma_O^2 = 1/M_Q$. 

When $O$ is Hermitian, e.g. the Hamiltonian $H$ in VQE, we can evaluate $X_{\theta,\theta'}$ with the Hadamard test incorporating Monte Carlo. First, we express the operator $O$ as a linear combination of unitary operators, $O = \sum_j a_j V_j$. Here, $a_j$ are complex coefficients, and $V_j$ are unitary operators. Then, we can express $X_{\theta,\theta'}$ as 
\begin{eqnarray}
X_{\theta,\theta'} = C_O \sum_j P_j \Re \left[e^{i\nu_j} \bra{\phi(\theta')}V_j\ket{\phi(\theta)}\right],
\end{eqnarray}
where $P_j = \abs{a_j}/C_O$ is the probability, $C_O = \sum_j \abs{a_j}$ is the normalisation factor, and $\nu_j = \gamma(\lambda,\theta)-\gamma(\lambda,\theta')+\arg a_j$. Finally, in each circuit shot, we randomly choose $V_j$ and $\nu_j$ according to the probability $P_j$. With the measurement outcomes, we evaluate $X_{\theta,\theta'}$ according to $X_{\theta,\theta'} = C_O \E[\mu]$. Therefore, the variance of $\hat{X}_{\theta,\theta'}$ has the upper bound $\sigma_O^2 = C_O^2/M_Q$. 

A similar approach applies to a general operator $O$. In general, $X_{\theta,\theta'}$ has real and imaginary parts, i.e. 
\begin{eqnarray}
X_{\theta,\theta'} &=& C_O \sum_j P_j \left\{\Re \left[e^{i\nu_j} \bra{\phi(\theta')}V_j\ket{\phi(\theta)}\right]\right. \notag \\
&&+i\Re \left.\left[e^{i(\nu_j-\pi/2)} \bra{\phi(\theta')}V_j\ket{\phi(\theta)}\right]\right\}.
\end{eqnarray}
We can measure the two parts separately. If we measure each part with $M_Q$ circuit shots, the variance of $\hat{X}_{\theta,\theta'}$ has the upper bound $\sigma_O^2 = 2C_O^2/M_Q$. 

\subsection{Ancilla-free Hadamard test}

For certain models, the Hadamard test can be implemented without the ancilla qubit, e.g.~fermion models with particle number conservation~\cite{obrien2021error,lu2021algorithms,xu2023quantum}. The ancilla-free circuit with particle number conservation is illustrated in Fig.~\ref{fig:circuits} (b). Suppose that the fermion system in states $\ket{\phi(\theta)}$ contains $m$ particle. If we take the Jordan-Wigner transformation for encoding fermions into qubits, qubit states $\ket{0}$ and $\ket{1}$ denote unoccupied and occupied fermion modes, respectively. Then, the number of ones is the same as the number of particles. Let's introduce an $m$-particle state as the initial state, $\ket{m} = \ket{1}^{\otimes m} \otimes \ket{0}^{\otimes (n-m)}$. Given this $m$-particle initial state, we can prepare $\ket{\phi(\theta)}$ via a particle-number-preserving transformation on $\ket{m}$. Let $\bar{U}(\theta) = U(\theta)\prod_{i=1}^m X_i$, where $X_i$ is the $X$ Pauli operator on the $i$th qubit. It can be verified that $\ket{\phi(\theta)} = \bar{U}(\theta)\ket{m}$. The ancilla-free Hadamard test works under the condition that all $\bar{U}(\theta)$ and $V_j$ are particle-number-preserving operators. In the final state, we have 
\begin{eqnarray}
\left\langle{Z\otimes\ketbra{0}{0}^{\otimes(n-1)}}\right\rangle = \Re \left[e^{i\nu_j} \bra{\phi(\theta’)}V_j\ket{\phi(\theta)}\right].
\end{eqnarray}
Accordingly, we measure all qubits in the computational basis. The measurement outcome takes three values $\mu = 0,\pm 1$: The outcome is the eigenvalue of the first-qubit $Z$ Pauli operator subjected to the condition that all other qubits are in the state $\ket{0}$; otherwise the outcome is zero. 

\section{Proof of Theorem~\ref{the:sign}}
\label{app:sign}

\begin{proof}
Suppose the sample size is $M$ for both $\hat{\mean{O}}$ and $\hat{\mean{\openone}}$. Let $\varepsilon_O=\sqrt{\frac{1}{M\kappa}(\norm{O}^2+\sigma_O^2)}$ and $\varepsilon_{\openone}=\sqrt{\frac{1}{M\kappa}(1+\sigma_{\openone}^2)}$. Then, using Eq.~(\ref{eq:variance}), we have $\Var\left(\hat{\mean{O}}\right) \leq \kappa \varepsilon_O^2$ and $\Var\left(\hat{\mean{\openone}}\right) \leq \kappa \varepsilon_{\openone}^2$. According to Chebyshev’s inequality, 
\begin{eqnarray}
\Pr(\abs{\delta_O}\leq \varepsilon_O)&\geq& 1-\kappa,\\
\Pr(\abs{\delta_{\openone}}\leq\varepsilon_{\openone})&\geq& 1-\kappa.
\end{eqnarray}
Therefore, 
\begin{eqnarray}
\Pr(\abs{\delta_O}\leq \varepsilon_O \text{ and } \abs{\delta_{\openone}}\leq\varepsilon_{\openone}) \geq 1-2\kappa.
\end{eqnarray}

When $M$ satisfies the condition in Eq. (\ref{eq:Mbound}), the following two inequalities hold: 
\begin{eqnarray}
\varepsilon_{\openone}&\leq& \mean{\openone},
\end{eqnarray}
\begin{eqnarray}
\frac{\mean{\openone}\varepsilon_O+\abs{\mean{O}}\varepsilon_{\openone}}{\mean{\openone}(\mean{\openone}-\varepsilon_{\openone})}&\leq& \varepsilon.
\end{eqnarray}
Then, when $\abs{\delta_O}\leq \varepsilon_O$ and $\abs{\delta_{\openone}}\leq\varepsilon_{\openone}$, 
\begin{eqnarray}
\left\vert\frac{\hat{\mean{O}}}{\hat{\mean{\openone}}} - \widetilde{\mean{O}}\right\vert &=& \left\vert\frac{\mean{O}+\delta_O}{\mean{\openone}+\delta_{\openone}} - \frac{\mean{O}}{\mean{\openone}}\right\vert \notag\\
&\leq&\frac{\mean{\openone}\abs{\delta_O}+\abs{\mean{O}}\abs{\delta_{\openone}}}{\mean{\openone}(\mean{\openone}-\abs{\delta_{\openone}})}\leq \varepsilon.
\end{eqnarray}
\end{proof}

\section{Proof of Theorem~\ref{the:variational}}
\label{app:variational}

\begin{proof}
According to Chebyshev's inequality, 
\begin{eqnarray}
\Pr\left(\abs{\hat{\mean{\openone}}-\mean{\openone}}\geq \kappa^{-\frac{1}{2}}\sqrt{\Var\left(\hat{\mean{\openone}}\right)} \right) \leq \kappa.
\end{eqnarray}
Because $\kappa^{-\frac{1}{2}}\sqrt{\Var\left(\hat{\mean{\openone}}\right)} \leq \eta_{\openone}$, 
\begin{eqnarray}
\Pr\left(\abs{\hat{\mean{\openone}}-\mean{\openone}}\geq \eta_{\openone}\right) \leq \kappa.
\end{eqnarray}
Similarly, 
\begin{eqnarray}
\Pr\left(\abs{\hat{\mean{H}}-\mean{H}}\geq \eta_H \right) \leq \kappa.
\end{eqnarray}
Therefore, 
\begin{eqnarray}
&&\Pr\left(\abs{\hat{\mean{\openone}}-\mean{\openone}} \leq \eta_{\openone} \text{ and } \abs{\hat{\mean{H}}-\mean{H}} \leq \eta_H \right) \notag \\
&\geq & 1-2\kappa.
\end{eqnarray}
In what follows, we focus on the case that $\abs{\hat{\mean{\openone}}-\mean{\openone}}\leq \eta_{\openone}$ and $\abs{\hat{\mean{H}}-\mean{H}}\leq \eta_H$ are true. 

Because $\mean{\openone}>0$, $\mean{\openone}+2\eta_{\openone}\geq \hat{\mean{\openone}}+\eta_{\openone}>0$. Therefore, 
\begin{eqnarray}
\frac{\mean{H}}{\hat{\mean{\openone}}+\eta_{\openone}} \leq \hat{L}'' \leq \frac{\mean{H}+2\eta _H}{\hat{\mean{\openone}}+\eta_{\openone}}.
\end{eqnarray}

When $\mean{H}\leq 0$, 
\begin{eqnarray}
\frac{\mean{H}}{\hat{\mean{\openone}}+\eta_{\openone}} \geq \widetilde{\mean{H}}.
\end{eqnarray}
When $\mean{H}> 0$, 
\begin{eqnarray}
\frac{\mean{H}}{\hat{\mean{\openone}}+\eta_{\openone}} > 0.
\end{eqnarray}
Therefore, 
\begin{eqnarray}
\hat{L}'' \geq \min\{\widetilde{\mean{H}},0\}.
\end{eqnarray}

When $\mean{H}\leq -2\eta_H$, 
\begin{eqnarray}
\frac{\mean{H}+2\eta_H}{\hat{\mean{\openone}}+\eta_{\openone}} &\leq& \frac{\mean{H}+2\eta_H}{\mean{\openone}+2\eta_{\openone}} \notag \\
&\leq& \widetilde{\mean{H}}+2\frac{\eta_H+\norm{H}\eta_{\openone}}{\mean{\openone}}.
\end{eqnarray}
Here, we have used that $\abs{\widetilde{\mean{H}}}\leq \norm{H}$. 
When $\mean{H}> -2\eta_H$, 
\begin{eqnarray}
\frac{\mean{H}+2\eta_H}{\hat{\mean{\openone}}+\eta_{\openone}} \leq \widetilde{\mean{H}}+\frac{2\eta_H}{\mean{\openone}}.
\end{eqnarray}
Therefore, 
\begin{eqnarray}
\hat{L}'' \leq \widetilde{\mean{H}}+2\frac{\eta_H+\norm{H}\eta_{\openone}}{\mean{\openone}}.
\end{eqnarray}
\end{proof}

\section{Covering number}
\label{app:covering}

The covering number has been introduced as an efficient measure for the expressivity of VQAs~\cite{du2022efficient}. Let $\mathcal{A}$ be a set of variational wavefunctions. Following the reference, we define the hypothesis space 
\begin{eqnarray}
\mathcal{H}(\mathcal{A}) = \left\{\left. \bra{\widetilde{\psi}}O\ket{\widetilde{\psi}} \,\right\vert \ket{\widetilde{\psi}}\in\mathcal{A}\right\}.
\end{eqnarray}
We use $\mathcal{N}(\mathcal{H},\epsilon,\norm{\cdot})$ to denote the covering number of $\mathcal{H}$ with the scale $\epsilon$ and norm $\norm{\cdot}$. Then, the larger $\mathcal{N}(\mathcal{H},\epsilon,\norm{\cdot})$ is, the more expressive $\mathcal{H}$ is. 

For two sets of variational wavefunctions $\mathcal{A}_1$ and $\mathcal{A}_2$, assuming $\mathcal{A}_1\subseteq \mathcal{A}_2$, it is obvious that $\mathcal{H}(\mathcal{A}_1)\subseteq \mathcal{H}(\mathcal{A}_2)$. Therefore, $\mathcal{N}(\mathcal{H}(\mathcal{A}_1),\epsilon,\norm{\cdot})\leq \mathcal{N}(\mathcal{H}(\mathcal{A}_2),\epsilon,\norm{\cdot})$. 


\section{Proof of Theorem~\ref{the:superiority}}
\label{app:superiority}

\begin{proof}
For all $\ket{\phi(\theta')},\ket{\phi(\theta)}\in\Omega$, 
\begin{eqnarray}
\norm{\ket{\phi(\theta')}-\ket{\phi(\theta)}}_2 \leq \xi\norm{\theta'-\theta}_2. 
\end{eqnarray}
Then, for the random-circuit variational wavefunction in Eq. (\ref{eq:DCpsi}), 
\begin{eqnarray}
&& \norm{\ket{\psi(\lambda)}-\ket{\phi(\theta)}}_2 \notag \\
&\leq & \left(\frac{\tau}{2}\right)^K \int d\theta' e^{-\tau\norm{\theta'-\theta}_1} \xi\norm{\theta'-\theta}_2 \notag \\
&\leq & \left(\frac{\tau}{2}\right)^K \int d\theta' e^{-\tau\norm{\theta'-\theta}_1} \xi\norm{\theta'-\theta}_1 \notag \\
&=& \frac{K\xi}{\tau}.
\end{eqnarray}
Because $\norm{\ket{\phi(\theta)}}_2 = 1$, 
\begin{eqnarray}
\mean{\openone}(\lambda) \geq \left(1-\frac{K\xi}{\tau}\right)^2,
\end{eqnarray}
i.e. $\ket{\widetilde{\psi}(\lambda)}\in\calV_x$, where $x = \left(1-\frac{K\xi}{\tau}\right)^2$. 

Taking the limit $\tau\rightarrow +\infty$, we have $\ket{\widetilde{\psi}(\lambda)}\in\calV_x$ for all $x<1$, and  
\begin{eqnarray}
&& \lim_{\tau\rightarrow +\infty}\norm{\ket{\psi(\lambda)}-\ket{\phi(\theta)}}_2 \notag \\
&=& \lim_{\tau\rightarrow +\infty}\norm{\ket{\widetilde{\psi}(\lambda)}-\ket{\phi(\theta)}}_2 = 0.
\end{eqnarray}

\end{proof}

\section{Proof of Theorem~\ref{the:UAT}}
\label{app:UAT}

\begin{proof}
Let $\{\ket{\phi(\theta_1)},\ket{\phi(\theta_2)},\ldots,\ket{\phi(\theta_d)}\}\subset\Omega$ be a basis of $\mathcal{H}$. Here, $d$ is the dimension of $\mathcal{H}$. Then, there exist coefficients $\beta_1,\beta_2,\ldots,\beta_d\in\mathbb{C}$ such that 
\begin{eqnarray}
\ket{\varphi} = \sum_{i=1}^d \beta_i \ket{\phi(\theta_i)}.
\end{eqnarray}
Accordingly, we can express $\ket{\varphi}$ as an integral, as given in Eq. (\ref{eq:varphi}), in which 
\begin{eqnarray}
\beta(\theta) = \sum_{i=1}^d \beta_i \delta(\theta-\theta_i).
\end{eqnarray}

Now, inspired by the wavefunction in Eq. (\ref{eq:DCpsi}), we approximate $\beta(\theta)$ with a continuous function 
\begin{eqnarray}
\beta_\tau(\theta) = \left(\frac{\tau}{2}\right)^K \sum_{i=1}^d \beta_i e^{-\tau\norm{\theta-\theta_i}_1}
\end{eqnarray}
The corresponding state is 
\begin{eqnarray}
\ket{\varphi_\tau} = \int d\theta \beta_\tau(\theta)\ket{\phi(\theta)} = \sum_{i=1}^d \beta_i \ket{\phi_\tau(\theta_i)}, 
\end{eqnarray}
where 
\begin{eqnarray}
\ket{\phi_\tau(\theta_i)} = \left(\frac{\tau}{2}\right)^K \int d\theta e^{-\tau\norm{\theta-\theta_i}_1}\ket{\phi(\theta)}.
\end{eqnarray}
This approximation introduces an error 
\begin{eqnarray}
\norm{\ket{\varphi_\tau}-\ket{\varphi}}_2 \leq \frac{K\xi}{\tau} \sum_{i=1}^d \abs{\beta_i}.
\end{eqnarray}
Here, we ave used that 
\begin{eqnarray}
\norm{\ket{\phi_\tau(\theta_i)}-\ket{\phi(\theta_i)}}_2 \leq \frac{K\xi}{\tau},
\end{eqnarray}
which can be proved as the same as in Appendix \ref{app:superiority}. 

To apply the universal approximation theorem~\cite{csaji2001approximation}, we map $\beta_\tau(\theta)$ to two real functions $\ln\abs{\beta_\tau(\theta)}$ and $\arg \beta_\tau(\theta)$. When $\beta_\tau(\theta)$ is continuous, $\ln\abs{\beta_\tau(\theta)}$ and $\arg \beta_\tau(\theta)$ are continuous functions. According to universal approximation theorem, for all $\epsilon_{AB}>0$, there exit $L\in\mathbb{N}$ and $\lambda\in\Lambda$ such that 
\begin{eqnarray}
\norm{A(\bullet)-\ln\abs{\beta_\tau(\bullet)}}_\infty \leq \epsilon_{AB}
\end{eqnarray}
and 
\begin{eqnarray}
\norm{B(\bullet)-\arg \beta_\tau(\bullet)}_\infty \leq \epsilon_{AB}.
\end{eqnarray}
Here, $\norm{f(\bullet)-g(\bullet)}_\infty = \sup_x \abs{f(x)-g(x)}$. With the $L$ and $\lambda$, we have 
\begin{eqnarray}
&& \norm{\alpha(\theta;\lambda)-\beta_\tau(\bullet)}_\infty \leq \epsilon_\beta,
\end{eqnarray}
where 
\begin{eqnarray}
\epsilon_\beta = \left(e^{\epsilon_{AB}}-1+\abs{e^{i\epsilon_{AB}}-1}\right) \sum_{i=1}^d \abs{\beta_i},
\end{eqnarray}
This error in the function leads to an error in the state 
\begin{eqnarray}
\norm{\ket{\psi(\lambda)}-\ket{\varphi_\tau}}_2 \leq \epsilon_\beta \int d\theta.
\end{eqnarray}

Notice that $\abs{\beta_i}$ and $\int d\theta$ ($\theta\in\Theta$, and $\Theta$ is compact) are finite. Therefore, there exists $\tau$ such that $\norm{\ket{\varphi_\tau}-\ket{\varphi}}_2 \leq \epsilon/2$, and there exist $L\in\mathbb{N}$ and $\lambda\in\Lambda$ such that $\norm{\ket{\psi(\lambda)}-\ket{\varphi_\tau}}_2 \leq \epsilon/2$. When $\ket{\psi(\lambda)}$ approaches to $\ket{\varphi}$, it is also normalised.
\end{proof}

\section{Prior guiding function}
\label{app:PGF}

We consider a prior guiding function related to the real-time evolution operator $e^{-iHtN_T}$. Before giving the prior guiding function, we introduce some relevant notations: i) $\mathcal{L}_k = \{\vec{l}=(l_1,l_2,\ldots,l_k)~\vert~l_q=1,2,\ldots,N_H\}$ is a set of $k$-dimensional vectors; ii) for a vector $\vec{l}\in \mathcal{L}_k$, 
\begin{eqnarray}
m(\vec{l},j) = \sum_{q=1}^k \delta_{l_q,j}
\end{eqnarray}
is the number of $j$ in components of $\vec{l}$; iii) 
\begin{eqnarray}
S(\vec{l}) = \sigma_{l_k}\cdots\sigma_{l_2}\sigma_{l_1}
\end{eqnarray}
and 
\begin{eqnarray}
T(\vec{l}) = \sigma_{N_{H}}^{m(\vec{l},N_H)}\cdots\sigma_{2}^{m(\vec{l},2)}\sigma_{1}^{m(\vec{l},1)}
\end{eqnarray}
are Pauli operators depending on the vector; and iv) 
\begin{eqnarray}
a(\vec{l}) = \frac{\Tr[S(\vec{l})T(\vec{l})^\dag]}{\Tr(\openone)}.
\end{eqnarray}
Because two Pauli operators are either commutative or anti-commutative, $S(\vec{l}) = a(\vec{l})T(\vec{l})$ and $a(\vec{l}) = \pm 1$. We remark that computing $m(\vec{l},j)$ and $a(\vec{l})$ on a classical computer is efficient, and the resource cost increases polynomially with the vector dimension $k$. 

The prior guiding function reads 
\begin{eqnarray}
F(\theta) = \prod_{i=1}^{N_T} f(t,\omega_i),
\label{eq:priorCF}
\end{eqnarray}
where 
\begin{widetext}
\begin{eqnarray}
f(t,\omega_i) = \prod_{j=1}^{N_H} \delta(\omega_{i,j}-h_jt)
+ \sum_{k=2}^\infty \sum_{\vec{l}\in \mathcal{L}_k} \frac{\prod_{q=1}^k(-ih_{l_q}t)}{k!}\left[a(\vec{l})-1\right]
\prod_{j=1}^{N_H}i^{m(\vec{l},j)} \delta\left(\omega_{i,j}-m(\vec{l},j)\frac{\pi}{2}\right).
\label{eq:f}
\end{eqnarray}
\end{widetext}

The property of the above prior guiding function is given by Eq. (\ref{eq:PGF}). We can prove it as follows. We only have to consider the integral over one $\omega_i$, 
\begin{eqnarray}
&& \int d\omega_i f(t,\omega_i) R(\omega_i) \notag \\
&=& e^{-i\sigma_{N_H}h_{N_H}t}\cdots e^{-i\sigma_{2}h_2t} e^{-i\sigma_{1}h_1t} \notag \\
&&+ \sum_{k=2}^\infty \sum_{\vec{l}\in \mathcal{L}_k} \frac{\prod_{q=1}^k(-ih_{l_q}t)}{k!}[a(\vec{l})-1]
T(\vec{l}).
\end{eqnarray}
The second line in the above equation corresponds to the first term of $f(t,\omega_i)$ in Eq.~(\ref{eq:f}), i.e.~taking $\omega_{i,j}=h_jt$ in the Hamiltonian ansatz. $T(\vec{l})$ corresponds to taking $\omega_{i,j} = m(\vec{l},j)\frac{\pi}{2}$: Taking this value of $\omega_{i,j}$, we have $R(\omega_i) = \prod_{j=1}^{N_H} i^{-m(\vec{l},j)}T(\vec{l})$, and the phase is cancelled by $\prod_{j=1}^{N_H} i^{m(\vec{l},j)}$ in $f(t,\omega_i)$. Using $S(\vec{l}) = a(\vec{l})T(\vec{l})$, we have 
\begin{eqnarray}
&& \int d\omega_i f(t,\omega_i) R(\omega_i) \notag \\
&=& e^{-i\sigma_{N_H}h_{N_H}t}\cdots e^{-i\sigma_{2}h_2t} e^{-i\sigma_{1}h_1t} \notag \\
&&+ \sum_{k=2}^\infty \sum_{\vec{l}\in \mathcal{L}_k} \frac{\prod_{q=1}^k(-ih_{l_q}t)}{k!}[S(\vec{l})-T(\vec{l})].
\end{eqnarray}
Then, applying the Taylor expansion, we have 
\begin{eqnarray}
&& e^{-i\sigma_{N_H}h_{N_H}t}\cdots e^{-i\sigma_{2}h_2t} e^{-i\sigma_{1}h_1t} \notag \\
&=& \openone + \sum_{k=1}^\infty \sum_{\vec{l}\in \mathcal{L}_k} \frac{\prod_{q=1}^k(-ih_{l_q}t)}{k!}T(\vec{l})
\end{eqnarray}
and 
\begin{eqnarray}
e^{-iHt}
= \openone + \sum_{k=1}^\infty \sum_{\vec{l}\in \mathcal{L}_k} \frac{\prod_{q=1}^k(-ih_{l_q}t)}{k!}S(\vec{l}).
\end{eqnarray}
Notice that $S(\vec{l})=T(\vec{l})$ for all $\vec{l}\in \mathcal{L}_1$. Therefore, 
\begin{eqnarray}
\int d\omega_i f(t,\omega_i) R(\omega_i) &=& e^{-iHt}.
\end{eqnarray}

\section{Proof of Lemma~\ref{lem:projection}}
\label{app:projection}

\begin{proof}
We consider the spectral decomposition of the Hamiltonian $H = \sum_{l=1}^{2^n} E_l\ketbra{\varphi_l}{\varphi_l}$. Here, $E_1\leq E_2\leq\cdots\leq E_{2^n}$ are eigenenergies in the ascending order, and $\ket{\varphi_l}$ are eigenstates. $E_1 = E_g$ is the ground-state energy, $\ket{\Phi_1}=\ket{\Psi_g}$ is the ground state, and $\Delta = E_2-E_1$ is the energy gap. Without loss of generality, we suppose $\ket{\Psi_0} = \sum_{l=1}^{2^n} \sqrt{p_l}\ket{\varphi_l}$, where $p_l$ is the probability of the eigenstate $\ket{\varphi_l}$, and $p_1 = p_g$. Then, 
\begin{eqnarray}
E(\lambda) &=& \frac{\sum_{l=1}^{2^n} e^{-(E_l-E_g)^2\tau^2} p_l E_l}{\sum_{l=1}^{2^n} e^{-(E_l-E_g)^2\tau^2} p_l} \notag \\
&=& E_g + \frac{\sum_{l=2}^{2^n} e^{-(E_l-E_g)^2\tau^2} p_l (E_l-E_g)}{p_g + \sum_{l=2}^{2^n} e^{-(E_l-E_g)^2\tau^2} p_l} \notag \\
&\leq & E_g + \frac{1-p_g}{p_g}Q,
\end{eqnarray}
where 
\begin{eqnarray}
Q = \max_{2\leq l\leq 2^n} e^{-(E_l-E_g)^2\tau^2} (E_l-E_g).
\end{eqnarray}

We derive two upper bounds of $Q$. First, notice that the maximum value of $e^{-x^2\tau^2}x$ is $e^{-\frac{1}{2}}(\sqrt{2}\tau)^{-1}$ at $x = (\sqrt{2}\tau)^{-1}$. Then, 
\begin{eqnarray}
e^{-(E_l-E_g)^2\tau^2} (E_l-E_g) \leq e^{-\frac{1}{2}}(\sqrt{2}\tau)^{-1}
\end{eqnarray}
for all $l$. Therefore, $Q\leq e^{-\frac{1}{2}}(\sqrt{2}\tau)^{-1}$. Second, we suppose that $\Delta\geq (\sqrt{2}\tau)^{-1}$. Under this condition, 
\begin{eqnarray}
Q = e^{-(E_2-E_g)^2\tau^2} (E_2-E_g) = e^{-\Delta^2\tau^2} \Delta.
\end{eqnarray}
\end{proof}

\section{Proof of Lemma~\ref{lem:cost}}
\label{app:cost}

\begin{proof}
The normalisation factor of the state in Eq.~(\ref{eq:HApsi}) is 
\begin{eqnarray}
\mean{\openone}(\lambda) &=& \left[\frac{\tau\sqrt{2\pi}}{C(\lambda)N_T}\right]^2 \bra{\Psi_0}e^{-(H-E_g)^2\tau^2}\ket{\Psi_0} \notag \\
&\geq & \left[\frac{\tau\sqrt{2\pi}}{C(\lambda)N_T}\right]^2 p_g.
\end{eqnarray}
Here, 
\begin{eqnarray}
C(\lambda) &=& \int d\theta e^{-\frac{N_T^2t^2}{2\tau^2}} \abs{F(\theta)} \notag \\
&=& \int dt e^{-\frac{N_T^2t^2}{2\tau^2}} \left[\int d\omega_i \abs{f(t,\omega_i)}\right]^{N_T}.
\end{eqnarray}
The contribution of each Trotter step is 
\begin{eqnarray}
\int d\omega_i \abs{f(t,\omega_i)} &\leq & 1 + 2\sum_{k=2}^\infty \sum_{\vec{l}\in \mathcal{L}_k} \frac{\prod_{q=1}^k(\abs{h_{l_q}t})}{k!} \notag \\
&=& 1 + 2(e^{h_{tot}\abs{t}}-1-h_{tot}\abs{t}) \notag \\
&=& 2e^{h_{tot}\abs{t}}-1-2h_{tot}\abs{t} \notag \\
&\leq & e^{2h_{tot}^2t^2}.
\end{eqnarray}
Therefore, 
\begin{eqnarray}
C(\lambda) &\leq & \int dt e^{-\frac{N_T^2t^2}{2\tau^2}} (2e^{h_{tot}\abs{t}}-1-2h_{tot}\abs{t})^{N_T} \notag \\
&=& \frac{\tau}{N_T} c(\tau).
\end{eqnarray}
When $N_T>4h_{tot}^2\tau^2$, 
\begin{eqnarray}
\frac{\tau}{N_T} c(\tau) &\leq & \int dt e^{-\frac{N_T^2t^2}{2\tau^2}} e^{2N_Th_{tot}^2t^2} \notag \\
&=& \frac{\tau\sqrt{2\pi}}{N_T}\frac{1}{\sqrt{1-\frac{4h_{tot}^2\tau^2}{N_T}}}.
\end{eqnarray}
Substituting the upper bound of $C(\lambda)$, we prove the lower bounds of $\mean{\openone}(\lambda)$. 

\end{proof}

\bibliography{reference}

\end{document}